\documentclass{amsart}
\usepackage{amsmath}
\usepackage{graphicx}
\usepackage{amsfonts}
\usepackage{amssymb}

\newtheorem{theorem}{Theorem}

\newtheorem{corollary}[theorem]{Corollary}

\newtheorem{lemma}[theorem]{Lemma}

\newtheorem{proposition}[theorem]{Proposition}
\newtheorem{remark}[theorem]{Remark}

\newcommand{\weakstar}{\stackrel{\ast}{\rightharpoonup}}
\newcommand{\weakst}{\stackrel{\ast}{\rightharpoonup}}

\newcommand{\weak}{\rightharpoonup}

\newcommand{\e}{\varepsilon} % epsilon
\newcommand{\en}{\varepsilon_n} % epsilon
\renewcommand{\phi}{\varphi} % phi
\newcommand{\Om}{\Omega} % Omega
\newcommand{\om}{\Omega} % Omega
\newcommand{\R}{\mathbb{R}} %R
\newcommand{\M}{\mathbb{M}}
\newcommand{\Md}{\mathbb{M}^{2\times 2}}
\newcommand{\Mdsym}{\mathbb{M}_\text{sym}^{2\times 2}}
\newcommand{\C}{\mathbb{C}}
\newcommand{\Z}{\mathbb{Z}}
\newcommand{\N}{\mathbb{N}}

\newcommand{\LM}[1]{\hbox{\vrule width.2pt \vbox to#1pt{\vfill \hrule
width#1pt
height.2pt}}}
\newcommand{\LL}{{\mathchoice
{\,\LM7\,}{\,\LM7\,}{\,\LM5\,}{\,\LM{3.35}\,}}}

\newcommand{\cE}{{\mathcal{E}}}

\newcommand{\Hh}{{\mathcal{H}}}

\newcommand{\Curl}{{\rm Curl\,}}
\newcommand{\curl}{{\rm curl\,}}
\newcommand{\Div}{{\rm Div\,}}
\newcommand{\bs}{\beta^{\rm sym}}
\newcommand{\sym}{^{\rm sym}}
\newcommand{\as}{^{\rm skew}}
\newcommand{\ms}{{\mathcal M}}
\newcommand{\asmn}{\mathcal{A}\mathcal{S}_{\e}}

\newcommand{\asm}{\mathcal{A}\mathcal{S}}
\newcommand{\res}{\mathop{\hbox{\vrule height 7pt width .5pt depth 0pt
\vrule height .5pt width 6pt depth 0pt}}\nolimits}

\newcommand{\f}{\varphi}

\newcommand{\F}{\mathcal{F}}

\newcommand{\Fdi}{\mathcal{F}^{\operatorname*{dilute}}}
\newcommand{\Fsu}{\mathcal{F}^{\operatorname*{super}}}

\title{Gradient theory for plasticity via homogenization of discrete dislocations}
\author[A. Garroni]
{Adriana Garroni}
\address[Adriana Garroni]{Dip. Mat. ``G. Castelnuovo'', Universit\`a di Roma ``La Sapienza'',
P.le Aldo Moro~5, I-00185 Roma, Italy.}
\email[A. Garroni]{garroni@mat.uniroma1.it}

\author[G. Leoni]
{Giovanni Leoni}
\address[Giovanni Leoni]{Department of Mathematical Sciences,
Carnegie Mellon University, Pittsburgh, PA, USA.}
\email[G. Leoni]{giovanni@andrew.cmu.edu }

\author[M. Ponsiglione]
{Marcello Ponsiglione}
\address[M. Ponsiglione]
{Dip. Mat. ``G. Castelnuovo'', Universit\`a di Roma ``La Sapienza'',
P.le Aldo Moro~5, I-00185 Roma, Italy.}
\email[M. Ponsiglione]{ponsigli@mat.uniroma1.it}

\begin{document}

\vskip .2truecm
\begin{abstract}
\small{In this paper, we deduce a macroscopic strain gradient theory for plasticity from a  model of discrete dislocations.
\par
We restrict our analysis to the case of a cylindrical symmetry for
the crystal in exam, so that the mathematical formulation will
involve a two dimensional variational problem.
\par
The dislocations are introduced as point topological defects of the strain fields, for which we compute
the elastic energy stored outside the so called core region.
We show that the $\Gamma$-limit as the core radius tends to zero and  the number of dislocations tends
to infinity of this energy (suitably rescaled), takes the form
$$
E= \int_{\Om} (W(\beta^e)  + \f (\Curl \beta^e)) \, dx,
$$
where $\beta^e$ represents the elastic part of the macroscopic strain, and $\Curl \beta^e$ represents the
geometrically necessary dislocation density.
The plastic energy density $\f$ is defined explicitly
through an asymptotic cell formula, depending only on the elastic tensor and the class of the admissible
Burgers vectors, accounting for the crystalline structure.
It turns out to be positively $1$-homogeneous, so that concentration on lines is permitted, accounting for
the presence of pattern formations observed in crystals such as  dislocation walls.

}
\end{abstract}

\maketitle
\tableofcontents

\section{Introduction}

In the last decade the study of crystal defects such as {\it dislocations} has been increasingly active.
The presence of dislocations (and their motion) is indeed considered the main mechanism of plastic deformations in metals.
Various phenomenological models have been proposed to account for  plastic effects due to dislocations,
such as the so called {\it strain gradient} theories.
\par
The aim of this paper is to provide a rigorous derivation of a strain gradient theory for plasticity as a
mesoscopic limit of systems of discrete dislocations, which are introduced  as point  defects of the strain fields,
for which we compute the elastic energy stored outside the so called {\it core region}.
\par
We  focus on the stored energy, disregarding dissipation, and  we
restrict our analysis to the case of a cylindrical symmetry for the
crystal in exam, so that the mathematical formulation will involve a
two dimensional variational problem.
\par
Many theories of plasticity are framed within  linearized
elasticity. In classical linear elasticity a displacement of $\om$
is a regular vector field $u:\om\to \R^2$. The equilibrium equations
have the form $\Div \C[e(u)]=0$, with $\C$ a linear operator from
$\R^{2\times2}$ into itself, and $e(u):=\frac {1}{2}(\nabla u
+(\nabla u)^{\top})$ the infinitesimal strain tensor. The
corresponding \emph{elastic energy} is
\begin{equation}
\int_{\om}W(\nabla u)\,dx \label{eps-functional0},
\end{equation}
where
\begin{equation}
W(\xi)=\frac{1}{2}  \C\xi:\xi \label{elastic density},
\end{equation}
is the \emph{elastic energy density}, and  the elasticity tensor $\C$ satisfies
\begin{equation}\label{coercon}
c_{1}|\xi^{\operatorname*{sym}}|^{2}\leq \C\xi:\xi \leq
c_{2}|\xi^{\operatorname*{sym}}|^{2} \qquad \text{ for any }
\xi\in\M^{2\times 2}\,,
\end{equation}
where $\xi^{\operatorname*{sym}}:=\frac12(\xi+\xi^{\top})$ and $c_1$
and $c_2$ are two given positive constants. In this linear framework
the presence of plastic deformations is classically modeled by the
additive decomposition of  the gradient of the displacement in the
elastic strain $\beta^e$ and the plastic strain $\beta^p$, i.e.,
$\nabla u=\beta^e +\beta^p$. The elastic energy induced by a given
plastic strain is then
\begin{equation}\label{pladist}
\int_\om W(\nabla u -\beta^p)\,dx\,.
\end{equation}
In view of their microscopic nature the presence of dislocations, responsible for the plastic deformation,
can be modeled in the continuum  by assigning the Curl of the field $\beta^p$. The quantity $\Curl \beta^p=\mu$
is then called the {\it Nye's dislocation density tensor},  (see \cite{Nye}).
Inspired by this idea, the strain gradient theory for plasticity, first proposed by Fleck and Hucthinson \cite{FH}
and then developed by Gurtin (see for instance \cite{G}, \cite{G2}), assigns an additional phenomenological energy
to the plastic deformation depending on the dislocation density, so that the stored energy looks like
\begin{equation}\label{pladist2}
\int_\om W(\nabla u -\beta^p)\,dx + \int_\om \varphi(\Curl \beta^p)\,dx\,.
\end{equation}
Note that for a given preexisting strain $\beta^p$, the minimum of the energy in \eqref{pladist} (possibly subject
to external loads) depends only on $\Curl \beta^p$. Hence the relevant variable in this problem is a strain field
$\beta$ whose Curl is prescribed.
The main issue with this model is the choice of the function
$\varphi$ in \eqref{pladist2}. In fact the usual choice (see for
instance \cite{GA}) is to take $\varphi$ quadratic, and then to fit
the free parameters through experiments as well as simulations, as
in \cite{NVG}. This choice has the well-known disadvantage of
preventing concentration of the dislocation density, that is instead allowed, for instance, by the choice $\f(\Curl\beta):=|\Curl \beta|$, recently proposed by  Conti and
Ortiz in \cite{CO}  (see also \cite{OR}).\\

The aim of this paper is to derive the strain gradient model \eqref{pladist2} for the stored energy, starting
from a basic model of discrete dislocations that accounts for the crystalline structure.  As a consequence
we will also get a  formula for the function $\varphi$, determined only by the elasticity tensor $\mathbb C$
and the Burgers vectors of the crystal.

A purely  discrete description of a crystal in the presence of dislocations can be given after introducing the
discrete equivalent of $\beta^p$ and  $\mu:=\Curl\beta^p$, following the approach  of  Ariza and Ortiz \cite{ArOr}.
That model is the starting point for a microscopic description of dislocations. In the passage from discrete to
continuum one can consider a intermediate (so to say) semi-discrete model known in literature as
{\it discrete dislocation model} (see also \cite{DNV}, \cite{CL}), in which the atomic scale is introduced as
an internal small parameter $\e$ referred to as the {\it core radius}.  The gap between the purely discrete
model and the discrete dislocation model can be usually filled through an interpolation procedure
(see  for instance \cite{Po}). In this paper,  to simplify matter,  we will adopt the discrete dislocation model.\\

In the discrete dislocation model a straight dislocation orthogonal to the cross section $\om$ is
identified with a point $x_0\in\om$ or, more precisely,  with a small region surrounding the dislocation
referred to as {\it core region}, i.e., a ball $B_\e(x_0)$,  being the core radius $\e$ proportional to
(the ratio between the dimensions of the crystal and) the underlying lattice spacing.
The presence of the dislocation  can be detected looking at the topological singularities of a continuum
strain field $\beta$, i.e.,
$$
\int_{\partial B_\e(x_0)}\beta \cdot t \, ds= \,b\,,
$$
where $t$ is the tangent to $\partial B_\e(x_0)$ and  $b$ is the Burgers vector of the given dislocation.

We will focus our analysis to the case when there may be many
dislocations crossing the domain $\om$. A generic distribution of
$N$ dislocations will be then identified with $N$ points
$\{x_i\}_{1,...,N}$ (or equivalently with the corresponding core
regions), each corresponding to some Burgers vectors $b_i$ belonging
to a finite set of {\it admissible Burgers vectors} $S$,  depending
on the crystalline structure, e.g., for square crystals, up to a
renormalization factor, $S= \{ e_1,\, e_2,\,-e_1,\,-e_2\}$.

In our analysis we introduce a second small scale $\rho_\e>\!\!\!> \e$ (the {\it hard core radius}) at which
a cluster of dislocations  will be identified with a multiple dislocation. In mathematical terms this corresponds
to introducing
the  span $\mathbb{S}$ of $S$ on $\Z$ (i.e., the set of finite combinations of Burgers vectors with integer coefficients)
and to representing
a generic distribution of dislocations as a measure $\mu$ of the type
$$
\mu=\sum_{i=1}^{N}  \delta_{x_i} \xi_i, \qquad \xi_i \in \mathbb{S}\,,
$$
where the distance between the $x_i$'s is at least $2\rho_\e$.
The admissible strain fields $\beta$ corresponding to $\mu$ are defined outside  the core region,
namely in $\Om\setminus \cup_i B_{\e}(x_i)$, and   satisfy
\begin{equation}\label{adco0}
\int_{\partial B_\e(x_i)}\beta\cdot t\, ds= \xi_i\,,\qquad \hbox{for all }\  \ i=1,...,N\,.
\end{equation}
It is well-known (and it will be discussed in detail in the sequel) that  the plastic distortion due
to the presence of dislocations  decays as the inverse of
the distance from the dislocations. This justifies the use of the linearized elastic energy outside the
core region $\cup_i B_{\e}(x_i)$. On the other hand, considerations at a discrete level show that the
elastic energy stored in the core region can be neglected.
Therefore, the elastic energy corresponding to $\mu$ and $\beta$ is given by
\begin{equation}\label{ee}
E_\e(\mu,\beta):=\int_{\om_\e(\mu)}W(\beta)\, dx\,,
\end{equation}
where $\om_\e(\mu)=\om\setminus \cup_i B_{\e}(x_i)$. By minimizing
the elastic energy \eqref{ee} among all strains satisfying
\eqref{adco0}, we obtain the energy induced by the dislocation
$\mu$. Note that, in view of the compatibility condition
\eqref{adco0}, this residual energy is positive whenever $\mu\neq
0$.
\par
In the case $W(\beta)=|\beta|^2$, with $\beta$ a curl free vector field in $\om\setminus \cup_i B_{\e}(x_i)$, this model coincides (setting $\beta=\nabla\theta$, $\theta$ being the phase function) with the energy proposed by Bethuel, Brezis and Helein (see \cite{BBH}) to deduce the so-called renormalized energy in the study of vortices in superconductors. 
In the context of dislocations this choice of $W$ and $\beta$ corresponds to a model for screw dislocations in an anti-planar setting. A $\Gamma$-convergence result in this framework
has been studied in \cite{Po} in  the energy regime $E_\e\approx |\log\e|$ that corresponds
to a finite number of dislocations.

In the planar context of edge dislocations, that we consider here, the connection with the Ginzburg-Landau model for vortices is more formal and the analysis has specific difficulties due to the vectorial nature of the circulation constraint (\ref{adco0}) and the properties of the energy density $W$.
A first analysis of this model has been perfomed by Cermelli and Leoni in  \cite{CL} where they obtain an asymptotic expansion of the energy functionals $E_\e$ in (\ref{ee}), in the case of a fixed configuration of dislocations. In particular they deduce the corresponding renormalized energy representing the effect of the Peach-Koehler force.
\\

The purpose of this paper is  to consider the energy functional
$E_\e$ in \eqref{ee}, where the unknowns are the distribution of
dislocations and the corresponding admissible  strains,  and to
study its asymptotic behavior in terms of $\Gamma$-convergence
%%Va detto molto meglio e motivato
as the lattice spacing $\e$ tends to zero and the number of
dislocations $N_\e$ goes to infinity. Under a suitable condition on
the hard core scale, we can show that in the asymptotics the energy
$E_\e$ can be decomposed into two effects: the {\it self-energy},
concentrated in the hard core region, and the {\it interaction
energy}, diffused in the remaining part of $\om$. The general idea
is that  the $\Gamma$-limit  of the (rescaled) energy functionals
$E_\e$ as $\e$ goes to zero is given by the sum of these two
effects. For different regimes for the number of dislocations $N_\e$
the corresponding $\Gamma$-limit will exhibit  the dominance of one
of the two effects: the self-energy is predominant for $N_\e <\!\!<
|\log\e|$, while the interaction energy is predominant for
$N_\e>\!\!>|\log \e|$.  In a critical regime $N_\e \approx |\log \e|
$ the two effects will be balanced and the structure of the limiting
energy is the following
\begin{equation}\label{Gammalimit}
\int_{\Om} W(\beta) \, dx + \int_{\Om} \varphi\left(\frac{d\mu}{d|\mu|}\right) \, d|\mu|.
\end{equation}
The first term  is the elastic energy of the limiting rescaled strain and comes from the interaction energy.
The second term represents the plastic  energy and comes from the self-energy; it depends only on the rescaled dislocation
density $\mu=\Curl\beta$ through a  positively homogeneous of degree $1$ density function $\f$,  defined by a suitable
asymptotic cell problem formula.
The constraint $\Curl \beta = \mu$ comes from the admissibility condition \eqref{adco0} that the
admissible strains have to satisfy at a discrete level.

If $(\beta, \mu)$ is a configuration that makes the continuous
energy in \eqref{Gammalimit} finite,  we have  necessarily that
$\beta\sym$ belongs to $L^2(\Om;\Md)$, and $\mu$ has finite mass. A
natural question arises: does $\mu$ belong to the space
$H^{-1}(\om;\R^2)$? We give a positive answer to this problem
proving a Korn type inequality for fields whose $\curl$ is a
prescribed measure, based on a fine estimate for elliptic systems
with $L\sp 1$-data recently proved by  Bourgain,  Brezis and  Van
Schaftingen (see \cite{BB}, \cite{BV}). In particular we deduce that
concentration of dislocations on lines is permitted, accounting for
the presence of pattern formations observed in crystals such as {\it
dislocation walls}, while concentration on points is not permitted.
An additional feature of the limit energy is the anisotropy of the
self-energy density inherited from the anisotropic elastic tensor
and the class of the admissible Burgers vectors accounting for the
crystalline structure.

The idea of deducing continuous models by homogenization of lower
dimensional singularities has been used in other contexts. Our result is very much
in the spirit of earlier results on the asymptotic analysis for the
Ginzburg-Landau model for superconductivity as the number of
vortices goes to infinity, done by Jerrard and Soner in \cite{JS},
and by Sandier and Serfaty in \cite{SS} (see also \cite{SS2} and the references therein).

As for the study of dislocations, a similar analysis was done by
Garroni and Focardi (\cite{GF}) starting from a phase field model introduced by  Koslowski, Cuitino
and Ortiz (\cite{KCO}) and inspired by the Peierls-Nabarro model (see also \cite{GM2}).

The plan of the paper is the following. In Section \ref{Scalings} we compute heuristically the
asymptotic behavior of the self and the interaction energy in terms of the number of dislocations  $N_\e$ and the core radius $\e$.
In Section \ref{setting}
we introduce rigorously the mathematical setting  of the problem, defining the class
of admissible configurations of dislocations, the class  of the corresponding admissible strains
and the rescaled energy functionals.
In Section \ref{cell} we provide the asymptotic cell problem formulas defining the density
of the plastic energy $\f$.
In Section \ref{Kor} we prove a Korn type inequality for fields with prescribed $\curl$.
In Section \ref{critical}
we give our main result concerning the $\Gamma$-convergence of the energy functionals in
the critical regime $N_\e\approx |\log\e|$, while
Section \ref{sub} and Section \ref{sup} are devoted to the sub critical and to the super critical
case respectively.

\section{Heuristic for the scaling}\label{Scalings}
In this section we identify the energy regimes, as the internal
scale $\e$ tends to zero, of the elastic energy induced by a
distribution of dislocations $\mu_\e:=\sum_{i=1}^{N_\e} \xi_i
\delta_{x_i} $, where we omit the dependence on $\e$ of the $x_i$'s.
The idea is that if the scaling of the  energy of a given sequence
of distributions of dislocations is assigned,  this provides a bound
for $N_\e$. In this section, in order to identify the relevant
energy regimes, we  do the converse: we assume to know the number of
dislocations $N_\e$ present in the crystal, and we  compute
(heuristically) the corresponding energy behavior as $\e$ goes to
zero.

As already mentioned in the introduction, the general idea is that
the energy will be always given by the sum of two terms: the {\it
self-energy}, concentrated in a small region (the {\it hard core
region}) surrounding the dislocations, and the {\it interaction
energy}, diffused in the remaining part of the domain $\Om$. We will
first introduce and compute heuristically the asymptotic behavior of
these two terms separately. Then we will  identify the leading term
(which depends on the behavior of $N_\e$) between these two
energies.

\subsection{Self-energy}
The self-energy of a distribution of dislocations $\mu_\e:=\sum_{i=1}^{N_\e} \xi_i \delta_{x_i}$ is
essentially given by the sum of the energies that would be induced by a single
dislocation $\mu_\e^i:= \xi_i \delta_{x_i}$, $ i=1, \ldots N_\e$  (i.e., if
no other dislocations  $\xi_j \delta_{x_j}, \, j\neq i$ were present in the crystal).
As we will see, the self-energy of a single dislocation $b \delta_{x}$ is asymptotically
(as $\e \to 0$) independent of the position of the dislocation point $x$ in $\Om$,
depending only on the elasticity tensor $\C$
and on the Burgers vector $b$.
\par
We begin by computing the self-energy in a very simple situation. We
assume for the time being that $\Om$ is the ball $B_1$ of radius one
and center zero,  that $\mu_\e\equiv \mu= b \, \delta_0 $, with
$|b|=1$, and we compute the energy of an admissible strain $\beta$
considering the toy energy given by the square of the $L^2$ norm of
$\beta$. More precisely, we consider
\begin{equation}\label{screw}
E_\e^{\text{toy}}(\mu,\beta):= \int_{B_1\setminus B_\e} |\beta|^2\,
dx\,,
\end{equation}
and the self-energy induced by $\mu$ given by minimizing this  energy among
all admissible strains $\beta$ compatible with $\mu_\e$, i.e., satisfying
$$
\int_{\partial B_\e}\beta \cdot t \, ds= b.
$$
In this setting the minimum strain $\beta(\mu)$ can be computed explicitly,
and in polar coordinates it  is given by the expression
$$
\beta(\mu)(\theta, r):= \frac{1}{2\pi r} t(r,\theta) \otimes b,
$$
where $t(r,\theta)$ is the tangent unit vector to $\partial B_r(0)$ at the point
$(r,\theta)$.
We deduce the following expression for the self-energy
\begin{equation}\label{a}
E^{\text{self}}_\e:= \int_\e^1 \frac{1}{2\pi r} \, dr =
\frac{1}{2\pi} (\log 1 - \log \e).
\end{equation}
Hence, as $\e \to 0$ the self-energy of a single dislocation behaves like
$|\log\e|$.
\par
Notice that most of the self-energy is concentrated around a small
region surrounding the dislocation. Indeed,  fix $s>0$, and  compute
the energy stored in the ball $B_{\e^s}(0)$,
$$
\int_{B_{\e^s}(0)} |\beta(\mu)|^2 \,dx=\int_\e^{\e^s} \frac{1}{2\pi
r} \, dr = \frac{1}{2\pi}(\log(\e^s) - \log\e)=
\frac{1}{2\pi}(1-s)|\log\e|=(1-s) E^{\text{self}}_\e.
$$
As we will see, in view of Korn inequality, the logarithmic behavior
and the concentration phenomenon of the self-energy  hold true also
in the case of elastic energies depending only on $\bs$ like in
\eqref{ee}. In this respect the position of the dislocation and the
shape of $\Om$ itself do not have  a big impact on the value of  the
self-energy. It seems then convenient to introduce the self-energy
as a quantity depending only on the Burgers vector $b$ (and the
elasticity tensor $\C$)  through a cell-problem. Before doing that,
we proceed heuristically, by introducing  the notion of  {\it hard
core} of the dislocation $\delta_0$, as a region surrounding the
dislocation such that
\begin{itemize}
\item[i)]
The hard core region contains almost all the self-energy;
\item[ii)]
The hard core region shrinks at the dislocation point as the atomic scale $\e$ tends to zero.
\end{itemize}
To this aim, in view of the previous discussion,  it is enough to define the  core
region as $B_{\rho_\e}(0)$, where the hard core radius $\rho_\e$ satisfies
\begin{itemize}
\item[i)]
$\lim_{\e\to 0} \frac{\rho_\e}{\e^s} = \infty$ for every fixed $0<s<1$;
\item[ii)]
$\rho_\e^2 \to 0$ as $\e\to 0$\,.
\end{itemize}
Note that condition $i)$ is indeed equivalent to
\begin{itemize}
\item[i')]
$\lim_{\e\to 0}\frac{\log(\rho_\e)}{\log\e} = 0 $.
\end{itemize}
A direct consequence of our definitions is that the self-energy  can
be identified (up to lower order terms) with the elastic energy
stored in the hard core region, for which, with a little abuse, we
will use the same notation, i.e.,
\begin{equation}\label{eself}
E_\e^{\text{self}}:=\int_{B_{\rho_\e}(0)} |\beta(\mu)|^2\,dx.
\end{equation}

Consider next the case of
a generic configuration of dislocations $\mu_\e:=\sum_{i=1}^{N_\e} b \, \delta_{x_i} $
in $\Om$ (but for sake of simplicity we will keep the Burgers vector equal to $b$ and the toy elastic energy \eqref{screw}).
The  hard core region of $\mu_\e$ is given by the union of the balls
$B_{\rho_\e}(x_i)$.
\par
Again we require a decay for $\rho_\e$ and that the area of the hard core region tends to zero as $\e\to 0$, $N_\e\to \infty$, i.e.,
\begin{itemize}
\item[i)]
$\lim_{\e\to 0} \frac{\rho_\e}{\e^s} = \infty$ for every fixed $0<s<1$;
\item[ii)]
$N_\e \rho_\e^2 \to 0$ as $\e\to 0$.
\end{itemize}
Moreover, we require that the dislocations are separated by a distance $2\rho_\e$, i.e.,
\begin{itemize}
\item[iii)] the balls $B_{\rho_\e}(x_i)$ are pairwise  disjoint.
\end{itemize}
This condition motivates the name of the hard core region.
In view of this assumption, it is natural to identify the self-energy of $\mu_\e$
as the elastic energy stored in the hard core region.
Therefore, we set
$$
E^{\text{self}}_\e(\mu_\e):=\int_{\text{Hard Core}}
|\beta(\mu_\e)|^2 \,dx= \sum_{i=1}^{N_\e} \int_{B_{\rho_\e}(x_i)}
|\beta(\mu_\e)|^2\,dx.
$$
We expect each term of the sum in the right hand side to be asymptotically equivalent to
the self-energy of a single dislocation  introduced in \eqref{eself}, so that in view of \eqref{a},
\begin{equation}
E_\e^{\text{self}}(\mu_\e)= \sum_{i=1}^{N_\e}
\int_{B_{\rho_\e}(x_i)} |\beta(\mu_\e)|^2\,dx
\, \approx \, N_\e|\log\e|.
\end{equation}
This expression represents the asymptotic behavior of the
self-energy as $\e\to 0$, $N_\e\to\infty$.

\subsection{Long range interaction between dislocations}
Since the self-energy is concentrated  in a small region surrounding
the dislocations, it is natural to define the  {\it interaction
energy} $E_\e^{\text{inter}}$ as the energy diffused in the
remaining part of the domain, far from each dislocation. In view of
this definition we have that the total energy
$E_\e(\mu_\e,\beta(\mu_\e))$ is given by the sum of the self and the
interaction energy:
\begin{align*}
E_\e(\mu_\e,\beta(\mu_\e))=\int_{\om_\e(\mu_\e)} W(\beta(\mu_\e))\,
dx= \int_{\text{Hard Core}}W(\beta(\mu_\e)\,dx \\+
\int_{\om\setminus \text{Hard Core}} W(\beta(\mu_\e))\,dx=
E_\e^{\text{self}}(\mu_\e) + E_\e^{\text{inter}}(\mu_\e)\,,
\end{align*}
where we recall that $\om_\e(\mu_\e)$ is the region of $\om$ outside
the dislocation cores, i.e.,
$\om_\e(\mu_\e):=\om\setminus\cup_{i=1}^{N_\e}B_\e(x_i)$, for
$\mu_\e=\sum_{i=1}^{N_\e}\delta_{x_i} b$. Let us compute
heuristically the interaction energy for a very simple configuration
of dislocations and for the toy energy considered in \eqref{screw}.
Let  $\Om$ be the unit ball $B_1$, and $\mu_\e$ be a configuration
of $N_\e$ periodically distributed dislocations whose Burgers vector
$b$ has modulus one. The stored interaction energy is given by
$$
E_\e^{\text{inter}}(\mu_\e,\beta(\mu_\e)):=\int_{\om\setminus
\text{Hard Core}} |\beta(\mu_\e)|^2 \,dx= \int_0^1 dr\,
\int_0^{2\pi r}\chi_{\e}|\beta(r,\theta)|^2 \, d\theta\,,
$$
where $\chi_{\e}$ denotes the characteristic function of the set
$\om\setminus \text{Hard Core}$. Thanks to the uniform distribution
of the $x_i$'s in $B_1=\om$ we deduce that the number of
dislocations contained in each ball of radius $r$ is proportional to
the area of such a ball, and more precisely is of order $\pi r^2
N_\e$. Therefore, the average $\overline \beta_t(r)$ of the
tangential component
of the strain  on each circle $\partial B_r$ outside the dislocation hard cores,
i.e., of $\beta_t(r,\theta)\chi_\e$, is of order  $\pi r^2 N_\e/2\pi r = N_\e r/2$.
The error is small thanks to the fact that the area of the hard cores region is negligible.
By Jensen's inequality we obtain the following estimate from above for the interaction energy
\begin{equation*}
E_\e^{\text{inter}} (\mu_\e) = \int_0^1  \int_0^{2\pi
r}\chi_\e|\beta(r,\theta)|^2 \, d\theta\, dr \ge \int_0^1   2\pi r
|\overline \beta_t(r)|^2 dr\approx \int_0^1   2\pi  N_\e^2 r^3/4
dr\,= C N_\e^2,
\end{equation*}
where $C$ is a constant independent of $\e$.
\par
Note that the self-energy has been estimated looking at the circulation condition that
the strain $\beta_\e(\mu)$ has to satisfy (in order to be an admissible strain) on
the circles $\partial B_r(x_i)$ for $\e\le r\le \rho_\e$, while the estimate for the
interaction energy has been established looking at the circulation condition on all circles $\partial B_r$ for $r \le 1$.
%We showed that the density of the interaction energy is bounded from below by $N_\e^2 r^2/4$.
We will see that this estimate is indeed sharp for recovery sequences in $\Gamma$-convergence.
We will then conclude that the interaction energy behaves  like $N_\e^2$ as $\e\to 0$, $N_\e\to\infty$.

\subsection{Energy regimes}
In view of the heuristic arguments of the previous section we fix a function $\e \mapsto N_\e$
that represents the number of dislocations present in the crystal corresponding to the internal scale $\e$.
The above  considerations can be summarized in three regimes for the behavior of $N_\e$ with respect to $\e\to 0$:
\begin{itemize}
\item[1.] {Dilute dislocations ($N_\e <\!\!< |\log\e|$)}: in this regime we have
that the self-energy, which is of order $N_\e |\log\e|$, is predominant with respect to the interaction energy.
\item[2.] {Critical regime ($N_\e \approx |\log\e|$)}: in this regime we have that the self-energy
and the interaction energy are both of order
$N_\e |\log\e|\approx |\log\e|^2$.
\item[3.] {Super-critical regime ($N_\e >\!\!> |\log\e|$)}: in this regime the interaction energy,
which is of order $|N_\e|^2$, is  predominant with respect to the self-energy.
\end{itemize}

\section{The setting of the problem}\label{setting}
In this section we specify the mathematical setting  of the problem.
In particular, we introduce the class $X_\e$ of admissible
configurations of dislocations, the class $\asmn(\mu)$ of the
corresponding admissible strains, and the rescaled energy
functionals $\F_\e$.
\par
From now on  $\Om$ is a bounded open subset of $\R^2$ with Lipschitz continuous boundary, representing
a horizontal section of an infinite cylindrical crystal. For the given crystal, we introduce the class
of Burgers vectors $S=\{b_1,...,b_s\}$. In what follows we assume that $S$ contains at least two (independent) vectors, so that
\begin{equation}\label{R2}
{\rm Span}_{\R} S=\R^2
\end{equation}
and this will imply that the function $\varphi$ in the energy
(\ref{Gammalimit}) is finite in whole of $\R^2$. The case of only
one Burgers vector is easier and it implies that $\varphi$ is finite
only on a one dimensional subspace of $\R^2$.
We denote by $\mathbb{S}$ the span of $S$ with integer coefficients ($\mathbb{S}={\rm  Span}_{\Z}S$),
i.e., the set of Burgers vectors for ``multiple dislocations".

As in Section \ref{Scalings}, $N_\e$  represents  the number of dislocations present in the crystal,
corresponding to the internal scale $\e$. We introduce also the sequence
$\rho_\e$  representing the radius of the hard core surrounding the dislocations, and we require
\begin{itemize}
\item[i)]
$\lim_{\e \to 0} \frac{\rho_\e}{\e_n^s} = \infty$ for every fixed $0<s<1$;
\item[ii)]
$\lim_{\e \to 0} |N_\e| \rho_\e^2 = 0$.
\end{itemize}
Condition i) says that the hard core region contains almost all the self-energy, while condition ii)
says that the area of the hard core region tends to zero, and hence that its complement contains almost
all the interaction energy. Note that conditions i) and ii) are compatible whenever

\begin{equation}\label{ipn}
N_\e \e^s \to 0  \qquad \text{ for every fixed } s>0.
\end{equation}

\par
We assume that the distance between any pair of dislocation points is at least $2\rho_\e$ and
we define the class  $X_{\e}$ of admissible dislocations by
\begin{multline}\label{xn}
X_{\e}:= \Big\{\mu\in \ms(\om;\R^2): \, \mu=\sum_{i=1}^M \xi_i \, \delta_{x_i}, \, M\in\N,\, B_{{\rho_\e}}(x_i) \subset \Om,\\
|x_j-x_k| \ge 2\rho_\e
\text{ for every }
j \neq k, \xi_i \in \mathbb{S}  \Big\}\,,
\end{multline}
where $\ms(\om;\R^2)$ denotes the space of vector valued Radon
measures. Given $\mu \in X_{\e}$ and $r\in\R$, we define
\begin{equation}\label{defoe}
\Om_{r}(\mu):= \Om \setminus \bigcup_{x_i \in \, \text{supp}(\mu)} \overline B_{{r}}(x_i),
\end{equation}
where $B_{{r}}(x_i)$ denotes the open ball of center $x_i$ and radius ${r}$.
\par
The class  of {\it admissible strains}
associated with any $\mu:= \sum_{i=1}^M \xi_i\, \delta_{x_i}\in X_\e $ is then
\begin{equation*}
\left\{
\beta \in L^2(\Om_\e(\mu);\M^{2\times 2}): \,
\Curl \beta = 0 \text{ in } \Om_{{\e}}(\mu)
\text{ and }
\int_{\partial B_\e(x_i)}\beta\cdot t\, ds= \xi_i\text{ for all} \
i=1,...,M\, \right\}.
\end{equation*}
Here $t$ denotes the tangent to $\partial B_{\e}(x_i)$
and the
integrand  $\beta \cdot t$ is intended in the sense of traces (see Theorem 2,
page 204 in \cite{DL} and Remark \ref{curlbeta} for more details).
In the arguments below it will be useful to extend the admissible strains to the whole of $\om$.
There are various extensions that can be considered and that are compatible with the model that
we have in mind. We decide to extend the $\beta$'s setting their value to be zero  in the dislocation cores.
Thus, from now on the class $\asmn(\mu)$ of {\it admissible strains}
associated with any $\mu:= \sum_{i=1}^M \xi_i \, \delta_{x_i} $ is given by
\begin{eqnarray}\label{adms}
\asmn(\mu):= \displaystyle\Big\{
\beta \in L^2(\Om;\M^{2\times 2}): \, \beta\equiv 0 \text { in } \Om\setminus \Om_{{\e}}(\mu),
\Curl \beta = 0 \text{ in } \Om_{{\e}}(\mu), \\
\nonumber \displaystyle \int_{\partial B_\e(x_i)}\beta\cdot t\, ds=
\xi_i,\ \text{and}\  \int_{\om_\e(\mu)}(\beta-\beta^T)\,dx=0 \text{
for all } i=1,...,M\, \Big\}.
\end{eqnarray}
In view of the definition of the elastic energy, the last condition in the definition of $\asmn(\mu)$
is not restrictive and it is there to guarantee uniqueness of the minimizing strain.

\begin{remark}\label{curlbeta}
{\rm
Let $\beta \in \asmn(\mu)$.
Formally $\Curl \beta$ is a tensor  defined by
$$
(\Curl \beta)_{i \, j \, l}:= \frac{\partial }{\partial {x_l}}\beta_{i \, j} - \frac{\partial }{\partial {x_j}}\beta_{i \, l},
$$
which by definition is antisymmetric with respect to the entries $j,\, l$. Therefore
$\curl \beta_{(i)}$ (where $\beta_{(i)}$ denotes the $i^{th}$ row of $\beta$) can be identified with the scalar
$$
\curl \beta_{(i)}:= \frac{\partial }{\partial x_1}\beta_{i \, 2}-\frac{\partial }{\partial x_2}\beta_{i \, 1}.
$$
In the sense of distribution, $\curl \beta_{(i)}$  is given by

\begin{equation}\label{disde}
<\curl \beta_{(i)}, \f> = \int_\Om <\beta_{(i)}, J \nabla \f>,
\end{equation}
where $J$ is the clockwise rotation of $90^{\text o}$.
From equation \eqref{disde} it turns out that whenever $\beta_{(i)}$ is in $L^2(\Om;\R^2)$,
then $\curl \beta_{(i)}$ is well defined for $\f\in H^1(\Om)$ and acts continuously on it; so that
$$
\curl \beta_{(i)}\in H^{-1}(\om;\R^2) \qquad \text{ for every } \beta\in \asmn(\mu).
$$
On the other hand, if
$\beta_{(i)}\in L^1(\Om;\R^2)$ and $\curl \beta_{(i)} \in H^{-1}(\om;\R^2)$
then $\beta_{(i)}$ is in $L^2(\Om;\R^2)$  modulo gradients.
\par
Finally,  notice that for every $\mu:=\sum_{i=1}^M \xi_i\, \delta_{x_i} \in X_{\e}$,
the  circulation condition in \eqref{adms} can be stated in the following equivalent way:
$$
<\Curl \beta, \f> = \sum_{i=1}^M   \xi_i \, c_i,
$$
for every $\f\in H_0^{1}(\Om)$ such that $\f\equiv c_i$ in
$B_{{\e}}(x_i)$. In particular, if $\f$ belongs also to
$C_0^0(\Om)$, we have
$$
<\Curl \beta, \f> \,  = \, \int_{\Om} \f \, d\mu.
$$
}
\end{remark}

\bigskip

The elastic energy $E_\e$ corresponding to a pair $(\mu,\beta)$, with $\mu\in X_{\e}$ and
$\beta \in \asmn(\mu)$, is defined  by
\begin{equation}\label{elaene}
E_\e(\mu,\beta):= \int_{\Om_{\e}(\mu)} W(\beta) \, dx\,,
\end{equation}
where
\begin{equation*}
W(\xi):=\frac{1}{2}  \C\xi:\xi
\end{equation*}
is the \emph{strain energy density}, and the elasticity tensor $\C$  satisfies
\begin{equation}\label{coercon2}
c_{1}|\xi^{\operatorname*{sym}}|^{2}\leq \C\xi:\xi \leq
c_{2}|\xi^{\operatorname*{sym}}|^{2} \qquad \text{ for any }
\xi\in\M^{2\times 2}\,,
\end{equation}
where $\xi^{\operatorname*{sym}}:=\frac12(\xi+\xi^{\top})$ and $c_1$
and $c_2$ are two given positive constants. Since $\beta$ is always
extended to zero outside $\om_\e(\mu)$, we can rewrite the energy as
$$
E_\e(\mu,\beta)= \int_{\Om} W(\beta) \, dx\,.
$$
Note that for any given $\mu\in X_\e$, the problem
$$
\min_{\beta\in \asmn(\mu)}\int_{\Om_\e(\mu)} W(\beta) \, dx\,
$$
has a unique solution. The proof of this fact can be obtained by taking of a finite number
of segments by $\Om_\e(\mu)$, obtaining in such a way a simply connected domain where to apply
Korn inequality to the curl-free admissible strains, and then following the direct method of calculus of variations.

We conclude giving the notation of the rescaled functionals  in the different regimes for $N_\e$,
according to the asymptotic analysis given in the previous section:
\begin{itemize}
\item[]{\it Sub-critical or dilute regime: $N_\e<\!\!< |\log\e|.$}
In this case the predominant contribution comes from the self-energy and is of  order $N_\e |\log{\e}|$. Therefore, we set
\begin{equation}\label{renefa1}
\Fdi_{\e}(\mu,\beta):=
\begin{cases}
\frac{1}{N_\e |\log{\e}|}  E_\e(\mu,\beta) & \text{ if } \mu\in X_{\e}, \, \beta \in \asmn(\mu),\\
\infty & \text{otherwise in }  L^2(\Om;\Md).
\end{cases}
\end{equation}
\item[]{\it Critical regime: $N_\e= |\log\e|.$}
In this case both energies are of  order $N_\e|\log{\e}|= |\log\e|^2$. Therefore, we set
\begin{equation}\label{renefa2}
\F_{\e}(\mu,\beta):=
\begin{cases}
\frac{1}{|\log{\e}|^2}  E_\e(\mu,\beta) & \text{ if } \mu\in X_{\e}, \, \beta \in \asmn(\mu),\\
\infty & \text{otherwise in }  L^2(\Om;\Md)
\end{cases}
\end{equation}
\item[]{\it Super-critical regime: $N_\e>\!\!> |\log\e|.$}
In this case the interaction  energy is the leading term, and is of  order $N_\e^2$. Therefore, we set
\begin{equation}\label{renefa3}
\Fsu_{\e}(\mu,\beta):=
\begin{cases}
\frac{1}{N_\e^2}  E_\e(\mu,\beta) & \text{ if } \mu\in X_{\e}, \, \beta \in \asmn(\mu),\\
\infty & \text{otherwise in }  L^2(\Om;\Md)
\end{cases}
\end{equation}
\end{itemize}

%By minimizing the  energy functionals $\Fdi_\e$, $\F_\e$ and $\Fsu_\e$  among
%all admissible strains compatible with a given distribution of dislocations $\mu$,
%we obtain the functionals $\cE\die(\mu):= \min_{\beta\in \asmn(\Om)}$ $\Fdi_
%\e(\mu,\beta)$,  $\cE_\e(\mu):= \min_{\beta\in \asmn(\Om)} \F_\e(\mu,\beta)
%$ and $\cE\sce(\mu) := \min_{\beta\in \asmn(\Om)} \Fsc_\e(\mu,\beta)$
%from $X_\e$ to $\R$
%representing  the energy stored in the crystal induced by the distribution of
%dislocations $\mu$.

\begin{remark}\label{regmis}
{\rm The choice of representing the Curl constraint by defining the
measure $\mu$ as a sum of Dirac masses concentrated in the
dislocation points is one of the possible choices. Other
possibilities would be to consider more regular measures with the
same mass, such as
\begin{equation}\label{diffusa}
\tilde\mu=\sum_{i=1}^{M} \frac{\chi_{B_\e(x_i)}}{\pi\e^2}\xi(x_i)
\end{equation}
or
\begin{equation}\label{sulbordo}
\hat\mu=\sum_{i=1}^{M} \frac{\Hh^1\res\partial B_\e(x_i)}{2\pi\e}\xi(x_i)\,.
\end{equation}
The advantage of these alternative choices is that $\tilde\mu$ and
$\hat\mu$  belong to $H^{-1}(\Om; \R^2)$, which is the natural space
for $\Curl\beta$. Indeed,
%by \eqref{disde},  in all regimes under consideration and for any given sequence
%$\{(\beta_\e,\mu_\e)\}$ in $\asmn(\mu_\e)\times X_\e$ with bounded energy
%we have (up to a subsequence) $\Curl \beta_\e/\lambda_\e \weakstar \Curl\beta
%$  in $H^{-1}(\Om; \R^2)$, where $\lambda_\e^2$ is equal to $N_\e|\log\e|$,
%$|\log\e|^2$ and $N_\e^2$ respectively in the three regimes,  and $\beta$ is
%the weak limit in $L^2$ of $\beta_\e/\lambda_\e$.
%Moreover, it can be proved that  the corresponding sequences $\{\tilde\mu_\e\}$
%and $\{\hat \mu_\e\}$ scaled by $\lambda_\e$ also converge weakly star in
%$H^{-1}$ to $\Curl\beta$.
%Finally,
since $\tilde \mu$ ($\hat \mu$ respectively) are in $H^{-1}(\Om;\R^2)$,
we could rewrite the class of admissible strains as follows:
\begin{equation*}
\left\{
\beta \in L^1(\Om;\M^{2\times 2}): \,
\Curl \beta = \tilde \mu \text{ (respectively $\hat \mu$) in } \Om \right\}.
\end{equation*}
This notion of admissible strains does not coincide with that given in \eqref{adms},
but it turns out to be equivalent to \eqref{adms}, in terms of $\Gamma$-convergence,
in the study of the asymptotic behavior of the energy functionals as $\e\to 0$.
}
\end{remark}

\section{Cell formula for the self-energy}\label{cell}
The self-energy stored in a neighborhood of a dislocation is of order $|\log \e|$ and,
in view of the concentration of the energy,  it is asymptotically  not affected
by the shape of the domain, depending only on the elasticity tensor $\C$ and on the Burgers vector $b$.
It seems then natural to introduce rigorously the notion of self-energy through a cell problem.
\par
In the following we will consider the self-energy of any  multiple Burgers vector $\xi \in \mathbb S$.
For convenience we will introduce all the quantities we need for a generic vector $\xi\in\R^2$.
\par
For every $\xi \in \R^2$ and for every $0<r_1< r_2 \in\R$, let
\begin{equation}\label{adcell}
\asm_{r_1,r_2}(\xi):= \left\{
\beta \in L^2(B_{r_2}\setminus B_{r_1};\M^{2\times 2}):
\Curl \beta = 0, \int_{\partial B_{r_1}}\beta \cdot t\, ds= \xi
\right\}\,,
\end{equation}
where $B_r$ denotes the ball of radius $r$ and center $0$.

We first note that the admissibility conditions above on a strain $\beta$, assure an a priori
bound from below for its energy. This is made precise by the following remark.

\begin{remark}\label{min-prob}
{\rm Given $0<r_1<r_2$ and $\xi\in\R^2$, for every admissible
configuration $\beta\in \asm_{r_1,r_2}(\xi)$ we have
$$
\int_{B_{r_2}\setminus B_{r_1}}|\beta\sym|^2 \, dx\geq c|\xi|^2\,,
$$
where the constant $c$ depends only on $r_1$ and $r_2$.
\par
Indeed, by introducing a cut with a segment $L$, the set $(B_{r_2}\setminus B_{r_1})\setminus L$ becomes simply connected. Since $\Curl \beta=0$
in $(B_{r_2}\setminus B_{r_1})\setminus L$, there exists a function $u\in H^1(B_{r_2}\setminus B_{r_1};\R^2)$ such that $\nabla u=\beta$
in $(B_{r_2}\setminus B_{r_1})\setminus L$. By the classical Korn's inequality applied to $u$ we  obtain
$$
\int_{B_{r_2}\setminus B_{r_1}} |\beta-A|^2\,dx\leq C \int_{B_{r_2}\setminus B_{r_1}}|\beta\sym|^2 \, dx,
$$
for some skew symmetric matrices  $A$. Moreover, by the fact that
$\beta\in \asm_{r_1,r_2}(\eta_n)$, we conclude
\begin{eqnarray*}
\int_{B_{r_2}\setminus B_{r_1}} |\beta-A|^2\,dx&\geq&\int_{r_1}^{r_2}\frac{{1}}{2\pi\rho}\left| \int_{\partial B_{\rho}}(\beta -A)\cdot t\, ds  \right|^2d\rho\\
&=&\int_{r_1}^{r_2}\frac{{1}}{2\pi\rho}\left| \int_{\partial B_{\rho}}\beta \cdot t\, ds  \right|^2d\rho=\int_{r_1}^{r_2}\frac{1}{2\pi\rho}|\xi|^2d\rho=\frac{|\xi|^2\log \frac{r_2}{r_1}}{2\pi}\,.
\end{eqnarray*}
}
\end{remark}

\bigskip

Set now $
C_\e:= B_1 \setminus  B_{\e},
$
and let $\psi_\e: \R^2\to \R$ be the function defined through the following cell problem
\begin{equation}\label{psie}
\psi_\e(\xi):= \frac{1}{|\log\e|}\min_{\beta\in \asm_{\e,1}(\xi)} \int_{C_\e}
W(\beta) \, dx \qquad \text{ for every } \xi \in \R^2.
\end{equation}

By \eqref{coercon2}, it is easy to see that problem \eqref{psie} has a solution.
%up to an additive antisymmetric matrix.
We will denote by $\beta_\e(\xi)$ the (unique) solution  of the cell problem \eqref{psie} whose average is a symmetric matrix.
Note that $\beta_\e(\xi)$ satisfies the boundary value problem
\begin{equation}\label{sdc}
\left\{
\begin{array}{lll}
\text{Div} \, \C \beta_\e(\xi) &= 0 & \text{ in } C_\e,\\
\C \beta_\e(\xi) \cdot \nu &= 0 & \text{ on } \partial C_\e.
\end{array}
\right.
\end{equation}

Moreover, we will denote by
$\beta_{\R^2}(\xi):\R^2 \to \R^2$  the planar strain defined in all $\R^2$ corresponding to the dislocation centered in $0$
with Burgers vector
$\xi$. The strain $\beta_{\R^2}(\xi)$ is of the type
\begin{equation}\label{betar2}
\beta_{\R^2}(\xi)(r,\theta)= \frac{1}{r} \Gamma_{\xi}(\theta),
\end{equation}
where the function $\Gamma$ depends on the elastic properties of the crystal, namely on the elasticity tensor $\C$,
and $\beta_{\R^2}(\xi)$ is a solution of the equation (we refer the reader to \cite{BBS} for a detailed treatment of the subject)
\begin{equation}\label{sdcp}
\left\{
\begin{array}{lll}
\Curl \beta_{\R^2}(\xi) &= \xi \delta_0  &\text{ in } \R^2,\\
\Div \C \beta_{\R^2}(\xi) &= 0  &\text{ in } \R^2.
\end{array}
\right.
\end{equation}

In the next lemma we will see that in the cell formula \eqref{psie}
we can assign  a suitable boundary condition without affecting the
asymptotic behavior of the energy  $\psi_\e(\xi)$.

\begin{lemma}\label{datoalb}
Let $\xi\in\R^2$ be fixed and let $ \hat \beta \in \asm_{1,\e}(\xi)$. Assume that $\hat \beta$ satisfies
\begin{equation}\label{decayeq0}
|\hat \beta(x)| \le \frac{K|\xi|}{|x|} \qquad \text{ for every } x\in C_\e \text{ and for some } K>0.
\end{equation}
Then there exists a function $\tilde \beta \in
\asm_{1,\e}(\xi)$ such that
\begin{itemize}
\item[1)] $\tilde \beta$ coincides with $ \hat \beta$ in a neighborhood of $\partial C_\e$;
\item[2)] $\displaystyle\int_{C_\e} W(\tilde \beta) \, dx \le \int_{C_\e} W(\beta_\e(\xi)) \, dx + C|\xi|^2$, where $C>0$
depends only on $K$.
\end{itemize}
\end{lemma}

\begin{proof}
As in Remark \ref{min-prob} we make the annulus  $C_\e$ a simply
connected domain by introducing a cut with a segment $L$ and we
denote it by $\tilde C_\e:= C_\e \setminus L$. Since both $\hat
\beta$ and $\beta_\e(\xi)$ are curl free in $\tilde C_\e$, we have
$$
\beta_\e(\xi) = \nabla u, \qquad \hat \beta = \nabla v \qquad \text{ for some } u,\, v\in H^1(\tilde C_\e;\R^2).
$$

We want to construct $\tilde \beta$ as the gradient of a convex combination
of $u$ and $v$. To this purpose,
let $\f_1:(\e,1/3)\to\R$ be a piecewise affine function defined by
$$
\f_1(r):=
\begin{cases}
0 & \text{ for } r \in   (\e, 2 \e),\\
1 & \text{ for } r \in   (3\e, 1/3),\\
\f_1 \text{ is linear } & \text{ in } (2\e,3\e)\\
\end{cases}
$$
and let $\f_2: [1/3,1)\to\R$ be defined by
$$
\f_2(r):=
\begin{cases}
1 & \text{ for } r = 1/3,\\
0 & \text{ for } r \in    [2/3, 1),\\
\f_2 \text{ is linear } & \text{ in } (1/3, 2/3).
\end{cases}
$$
Define
$$
C_\e^1:= \{x\in C_\e: 2\e\le |x| \le 3\e \}, \qquad C_\e^2:=
\{x\in C_\e: 1/3 \le |x| \le 2/3 \}.
$$
Finally, set
$$
\tilde \beta(x):=
\begin{cases}
\nabla \big( \f_1(|x|) (u(x)-c_1(u)) + (1-\f_1(|x|)) (v(x)-c_1(v))\big) & \text{ for } |x|\le 1/3,\\
\nabla \big( \f_2(|x|) (u(x)-c_2(u)) + (1-\f_2(|x|)) (v(x)-c_2(v)) \big)
& \text{ for } |x| > 1/3,
\end{cases}
$$
where $c_i(f)$ denotes the mean value of the function $f$ on
$C_\e^i$. It is easy to check that by construction $\tilde \beta$
belongs to $\asm_{1,\e}(\xi)$ and coincides with $\hat\beta$ in a
neighborhood of $\partial C_\e$. It remains to prove property $2)$.
\par
By construction we have that $\tilde \beta$ coincides  with $\beta_\e(\xi)$  in the region
$$
C_\e^3:= \{x\in C_\e: 3\e \le |x| \le 1/3 \}.
$$
Therefore, to conclude the proof of property $2)$ it is enough to compute the energy stored in  $C_\e \setminus C_\e^3$.

In view of \eqref{coercon2}, of the fact that $||\nabla
\f_1||_\infty \le C/\e$, and of the following Poincar\'e inequality
$$
\int_{(B_{3r}\setminus B_{2r})\setminus L} |u-c(u)|^2 \, dx \le c r^2
\int_{(B_{3r}\setminus B_{2r})\setminus L} |\nabla u|^2 \, dx,
$$
where $r>0 $ and $c(u)$ stands for the average of $u$ over the
domain, it is easy to check that property $2)$ holds provided the
following estimates are established
\begin{equation}\label{fole}
\int_{C_\e\setminus C_\e^3} |\nabla u|^2 \,dx\le C|\xi|^2, \qquad
\int_{C_\e\setminus C_\e^3} |\nabla v|^2 \,dx\le C|\xi|^2,
\end{equation}
for some $C$ independent of $\e$.
\par
It remains to prove \eqref{fole}. By \eqref{decayeq0} it
follows that
\begin{equation}\label{agle}
\int_{C_\e\setminus C_\e^3} |\nabla v|^2 \, dx \le C |\xi|^2\left (
\int_\e^{2\e} \frac{1}{t} \,dt + \int_{2/3}^1 \frac{1}{t} \,dt
\right) \le C|\xi|^2.
\end{equation}

Concerning $u$, in view of \eqref{sdc} and \eqref{sdcp}, we have that
$\nabla u$ can be written as
\begin{equation}\label{cbw}
\nabla u= \beta_{\R^2}(\xi) + \nabla h,
\end{equation}
with $h$ given by
$$
\nabla h(x)= -  \int_{\partial B_1(0) \cup \partial B_{\e}(0)}
\nabla G(x,y) \,  \C \beta_{\R^2}(\xi) \cdot \nu(y) \, dy \qquad
\text{ for } x\in \tilde C_\e.
$$
Here $G(x,y)$ is the Green function corresponding to the elasticity tensor $\C$,
satisfying the equation
$$
-\Div_x \C \nabla_x G(x,y) = \delta_{y} I \qquad \text{ in } \R^2
$$
for every fixed $y$. It is well-known (see \cite{BBS}) that

\begin{equation}\label{stime}
|\nabla G(x)| \le C/|x| \quad\text{ for every }\quad x\in\R^2\,,
\end{equation}
for some constant $C>0$.
From \eqref{cbw} and \eqref{stime} it follows as in \eqref{agle} that
$$
\int_{C_\e\setminus C_\e^3} |\nabla u|^2 dx\le 2 \int_{C_\e\setminus
C_\e^3} \left(|\beta_{\R^2}(\xi)|^2 + |\nabla h|^2\right)\,dx \le
C|\xi|^2,
$$
which concludes the proof of \eqref{fole} and therefore of the lemma.
\end{proof}

From Lemma \ref{datoalb}, together with (\ref{betar2}), we deduce
the following corollary

\begin{corollary}\label{coro}
There exists a  constant $C>0$ such that for every $\xi\in \R^2$,
\begin{equation}\label{}
\psi_\e(\xi) \le \frac{1}{|\log\e|} \int_{C_\e}
W(\beta_{\R^2}(\xi))\,dx \le \psi_\e(\xi) + \frac{C
|\xi|^2}{|\log\e|}.
\end{equation}
In particular,  as $\e\to 0$ the functions $\psi_\e$ converge pointwise  to the function $\psi: \R^2\to \R$, defined by
\begin{equation}\label{pcsf}
\psi(\xi):=\lim_{\e\to 0 }\psi_\e(\xi) =
\frac{1}{|\log\e|}\lim_{\e\to 0 }\int_{C_\e} W(\beta_{\R^2}(\xi))
\,dx= \int_{\partial B_1(0)} W(\Gamma_{\xi}(\theta)) \, d\theta.
\end{equation}
More precisely,  we have
\begin{equation}\label{mip}
|\psi_\e(\xi)- \psi(\xi)| \le \frac{C|\xi|^2}{|\log \e|}.
\end{equation}

\end{corollary}

By means of a simple change of variable $\e\to \rho$ we have that the self-energy
is indeed concentrated in a $\rho$-neighborhood of the dislocation points whenever
$|\log\rho| <\!\!< |\log \e|$. The precise statement is given in the next proposition.

\begin{proposition}\label{equicell}
For every $\e>0$ let $\rho_\e>0$ be such that $\log \rho_\e/\log \e\to 0$
as $\e \to 0$. Let
$\bar\psi_{\e}: \R^2\to \R$ be defined through the following minimum problem
\begin{equation}\label{rhopsie0}
\bar \psi_{\e}(\xi):= \frac{1}{|\log\e|}
\min \Big\{
\int_{B_{\rho_\e}\setminus B_\e}
W(\beta) \, dx:
\beta \in \asm_{\e,\rho_\e}(\xi)  \Big\}.
\end{equation}
Then $\bar \psi_{\e} = \psi_{\e} (1 + o(\e))$, where $o(\e)\to 0$ as
$\e\to 0$ uniformly with respect to $\xi$. In particular, $\bar
\psi_{\e}$ converge pointwise as $\e\to 0$ to the  function $\psi:
\R^2\to \R$ given in Corollary \ref{coro}.
\par
Moreover, let  $\hat \beta \in \asm_{\e,\rho_\e}(\xi)$, $\xi\in\R^2$,  be such that
\begin{equation}\label{pallino}
|\hat \beta(x)|\le K\frac{|\xi|^2}{|x|}
\end{equation}
for some $K\in\R$, and let
$\tilde\psi_{\e}: \R^2\to \R$ be defined through the following minimum problem
\begin{equation}\label{rhopsie}
\tilde \psi_{\e}(\xi):= \frac{1}{|\log\e|}
\min \Big\{
\int_{B_{\rho_\e}\setminus B_\e}
W(\beta) \, dx:
\beta \in \asm_{\e,\rho_\e}(\xi), \,  \beta \cdot t = \hat \beta \cdot t \text{ on } \partial B_\e \cup \partial B_{\rho_\e} \Big\}.
\end{equation}
Then $\tilde \psi_{\e} = \psi_{\e} (1 + o(\e))$, where $o(\e)\to 0$
as $\e\to 0$  uniformly with respect to $\xi$. In particular,
$\tilde \psi_{\e}$ converge pointwise as $\e\to 0$ to the  function
$\psi: \R^2\to \R$ given in Corollary \ref{coro}.
\end{proposition}

\begin{remark}\label{mp}
{\rm
The error $o(\e)$ appearing in the expression of $\bar \psi$ and $\tilde \psi$ in Proposition \eqref{equicell}
can be estimated as follows
$o(\e)\approx \log \rho_\e/ \log \e$.
}
\end{remark}
We are now in a position to define the density $\varphi:\R^2\to[0,\infty)$ of the self-energy through
the following relaxation procedure
\begin{equation}\label{defi}
\varphi(\xi):= \inf \left \{ \sum_{k=1}^N \lambda_k \psi(\xi_k)\,: \ \sum_{k=1}^N \lambda_k
\xi_k = \xi\,,\  N\in\N\,, \ \lambda_k\geq 0\,, \ \xi_k\in \mathbb S \right \}.
\end{equation}
It follows from the definition that the function $\f$ is positively $1$-homogeneous and convex.
Moreover, since $\psi(\xi)\ge C|\xi|^2$ for some $C\ge 0$, (that can be checked by its very definition),
the $\inf$ in \eqref{defi} is indeed a minimum.

\begin{remark}\label{burgersR}
{\rm
Note that if for every $z_1, \, \ldots , z_s \in \Z$ we have
\begin{equation}\label{burgers}
\psi\left(\sum_{i=1}^s z_i b_i\right) \ge \sum_{i=1}^s z_i \psi(b_i),
\end{equation}
then in the relaxation procedure given in  \eqref{defi} we can replace $S$ with $\mathbb S$. More precisely,
the formula for $\f$ reduces to
\begin{equation}\label{defi2}
\varphi(\xi):= \min \left\{ \sum_{i=1}^s |\lambda_i| \psi(b_i): \,
\sum_{i=1}^s \lambda_i b_i = \xi, \, b_i \in {S} \right\}.
\end{equation}
Actually, condition \eqref{burgers} can be viewed as a condition that the class of Burgers vectors $b_i$
of the given crystal has to satisfy in order to contain all the dislocation's defects observed in the crystal.
In other words, if a dislocation
corresponding to a vector $b:=\sum_{i=1}^s z_i b_i$  stores an energy smaller than that obtained by separating
all its components $b_i$, then $b$ itself has to be considered as a Burgers vector of the crystal.
A rigorous mathematical  definition of the class of Burgers vectors corresponding to a given crystal could be
to consider the set $\{b_1, \ldots b_s\}\subset \mathbb S$ of all vectors satisfying $\psi(b_i)= \f(b_i)$,
where $\mathbb S$ is the set of slips under which the crystal is invariant. The Burgers vectors defined in such a
way would always satisfy property \eqref{burgers}.
%Then  condition \eqref{burgers} can be understood as a rigorous mathematical
%definition of the class of Burgers vectors corresponding to a given crystal.
}
\end{remark}

\section{A Korn type inequality for fields with prescribed curl}\label{Kor}
Let $u\in W^{1,2}(\Om;\R^2)$, $\beta$ its gradient and  $\beta\sym$ and $\beta\as$ its decomposition in symmetric
and anti-symmetric part. The classical Korn inequality asserts that if $\beta\as$ has zero mean value, then
its $L^2$ norm is controlled by the $L^2$ norm of $\beta\sym$. We will show that in dimension two the same
result holds true also for fields $\beta$ that are not curl free, modulo an error depending actually on
the mass of $\Curl \beta$. The result is a consequence of some estimates for elliptic systems with $L\sp 1$-data
recently proved by  Bourgain,  Brezis and  Van Schaftingen in \cite{BB}, \cite{BV}.
The precise statement is the following
\begin{theorem}[A Generalized Korn type inequality]\label{KTI}
There exists a constant $C$ depending only on $\Om$ such that for every  $\beta\in L^1(\Om;\Md)$  with
$$
\Curl \beta=\mu \in \ms(\Om;\R^2),  \qquad  \int_\Om \beta\as=0,
$$
we have
\begin{equation}\label{KTIE}
\int_{\Om} |\beta\as|^2 \, dx \le C\left( \int_{\Om} |\beta\sym|^2 \, dx + \big(|\mu|(\Om)\big)^2 \right).
\end{equation}
\end{theorem}

\begin{proof}
The condition $\Curl \beta=(\mu_1,\mu_2)$ can be written in the following form

$$
\left\{
\begin{array}{ll}
(\beta\as_{12})_{x_1} = & h_1 + \mu_1, \\
(\beta\as_{12})_{x_2} = & h_2 + \mu_2,
\end{array}
\right.
$$
where $h_i\in H^{-1}(\Om)$ are linear combinations of derivatives of entries of $\beta\sym$.

Since the field $\big ((\beta\as_{12})_{x_1}, (\beta\as_{12})_{x_2}\big)$ is curl free,
we deduce that $\curl (\mu_1,\mu_2)= \curl(-h_1,-h_2)$, or equivalently
\begin{equation}\label{div}
\text{div } (-\mu_2, \mu_1) = \text{div } (h_2, -h_1).
\end{equation}
By \cite[Lemma 3.3 and Remark 3.3]{BV} (see also \cite{V} and \cite{BB}) we have that if $f\in L^1(\Omega; \R^2)$
is a vector field satisfying ${\rm div}\, f\in H^{-2}(\Omega)$, then $f$ also belong to $H^{-1}(\Omega)$ and
the following estimate also holds
$$
\|f\|_{H^{-1}(\Om)} \le c \big( \|\text{div } f\|_{H^{-2}(\Om)}
+\|f\|_{L^1(\Omega)}  \big)\,.
$$
This result clearly extents by density to measures with bounded variations. Thus, by (\ref{div}),
we have that $\text{div } (-\mu_2, \mu_1)\in H^{-2}(\Omega)$, and so we deduce that $\mu$ belongs to $H^{-1}(\Omega)$ and
\begin{equation}\label{BVesti}
\|\mu\|_{H^{-1}(\Om)} \le c \big( \|\text{div } (h_2, -h_1)\|_{H^{-2}(\Om)} +|\mu|(\Om)  \big) \le c \big( \|\beta\sym\|_{L^{2}(\Om;\R^2)} +|\mu|(\Om) \big)\,.
\end{equation}

Now let $u\in H^1_0(\Om;\R^2)$ be the solution of $-\Delta
u=(-\mu_2,\mu_1)$ in $\Om$ and let $\xi$ be the $2\times 2$ matrix
defined by $\xi:=J\nabla u$ (i.e., the $i^{th}$ row of $\xi$ is
given by $\xi_i=(-(u_i)_{x_2}, (u_i)_{x_1})$, for $i=1,2$). By
definition we have that $\Curl \xi= \mu$.  By \eqref{BVesti} we then
obtain
\begin{equation}\label{xile}
\int_{\Om} |\xi|^2 \, dx = \|\nabla u\|^2_{L^2}\leq
c\|\mu\|_{H^{-1}(\Om)} \leq  \left(\int_{\Om} |\beta\sym|^2 \, dx +
\big(|\mu|(\Om)\big)^2 \right).
\end{equation}
Since the average of the anti-symmetric part of $\xi$ can be easily
estimated with the $L^2$ norm of $\xi$, we can assume that
(\ref{xile}) holds for a matrix $\xi$ such that $\xi\as$ has zero
mean value and $\Curl \xi=\mu$. Therefore, by the  classical Korn
inequality applied to $\beta - \xi$, which by construction is curl
free, and in view of \eqref{xile} we conclude
\begin{align*}
\int_{\Om} |\beta\as|^2 \, dx &\le c\left(\int_{\Om}
|\beta\as-\xi\as|^2 \, dx + \int_{\Om} |\xi\as|^2 \, dx \right)\\
&\le c\left(\int_{\Om} |\beta\sym-\xi\sym|^2
\, dx + \int_{\Om} |\xi\as|^2 \, dx \right)\\
&\le c\left( \int_{\Om} |\beta\sym|^2  \, dx +  \int_{\Om} |\xi|^2
\, dx  \right) \le c\left( \int_{\Om} |\beta\sym|^2 \, dx +
\big(|\mu|(\Om)\big)^2 \right).
\end{align*}
\end{proof}

\section{The critical regime ($N_\e \approx |\log\e|$)}\label{critical}
In this section we will study the asymptotic behavior of the
rescaled energy functionals  as the internal scale $\e \to 0$, in
the critical energy regime, namely with $N_\e=|\log \e|$. In terms
of $\Gamma$-convergence, we consider the rescaled energy functionals
$\F_\e:\ms(\om;\R^2)\times L^1(\Om;\Md)\to \R$ defined in
\eqref{renefa2}.
\par
According to the heuristic arguments above, in this regime we aspect  the coexistence of the two effects,
the interaction energy and the self-energy, so that the candidate for the  $\Gamma$-limit $\F:
\ms(\om;\R^2)\times L^2(\Om;\M^{2\times 2})\to \R$ is defined by
\begin{equation}\label{gali}
\F(\mu,\beta):=
\begin{cases}\displaystyle
\int_{\Om} W(\beta) \, dx + \int_{\Om} \varphi\left(\frac{d\mu}{d|\mu|}\right) \, d|\mu| & \text{ if } \mu\in H^{-1}(\om;\R^2), \, \Curl \, \beta = \mu;\\
\infty & \text{ otherwise in $L^2(\Om;\Md)$.}
\end{cases}
\end{equation}

\begin{theorem}\label{mainthm}
The following $\Gamma$-convergence result holds.
\begin{itemize}
\item[i)]{\bf Compactness.}
Let $\en\to 0$ and let $\{(\mu_n,\beta_n)\}$ be a sequence in
$\ms(\om;\R^2)\times L^2(\Om;\M^{2\times 2})$ such that
$\F_{\en}(\mu_n,\beta_n)\le  E$ for some positive constant $ E$
independent of $n$. Then there exist a subsequence of $\e_n$ (not
relabeled),  a measure $\mu\in H^{-1}(\om;\R^2)$, and a strain
$\beta\in L^2(\Om;\M^{2\times 2})$, with $\Curl \beta = \mu$, such
that
\begin{equation}\label{not1}
\frac{1}{|\log {\en}|}\mu_n \weakst \mu \qquad \hbox{in}\  \ms(\om;\R^2)\,,
\end{equation}
\begin{equation}\label{not2}
\frac{1}{|\log {\en}|} \beta_n \weak \beta \qquad \hbox{in}\
L^2(\Om;\M^{2\times 2})\,.
\end{equation}

\item[ii)]{\bf $\Gamma$-convergence.}
The functionals  $\F_{\e}$ $\Gamma$-converge to $\F$ as ${\e} \to 0$ with respect to
the convergence in \eqref{not1},  \eqref{not2}, i.e., the following inequalities hold.
\medskip

\begin{itemize}
\item[]{\bf $\Gamma$-liminf inequality:}
for every  $(\mu,\beta) \in \left(\ms(\om;\R^2)\cap H^{-1}(\om;\R^2)\right) \times L^2(\Om;\M^{2\times 2})$,
with $\Curl \beta = \mu$, and for every sequence
$(\mu_\e,\beta_\e) \in X_{\e}\times L^2(\Om;\M^{2\times 2})$
satisfying \eqref{not1} and \eqref{not2}, we have
$$
\liminf_{\e\to 0}\F_{\e}(\mu_\e,\beta_\e) \ge \F(\mu,\beta).
$$
\item[]{\bf $\Gamma$-limsup inequality:}
given $(\mu,\beta)\in \left(\ms(\om;\R^2)\cap H^{-1}(\om;\R^2)\right) \times L^2(\Om;\M^{2\times 2})$,
with $\Curl \beta = \mu$, there exists
$(\mu_\e,\beta_\e)\in \ms(\om;\R^2)\times L^2(\Om;\M^{2\times 2})$ satisfying \eqref{not1} and \eqref{not2},  such that
$$
\limsup_{\e\to 0}\F_{\e}(\mu_\e,\beta_\e) \le \F(\mu,\beta).
$$
\end{itemize}
\end{itemize}
\end{theorem}

\begin{remark}\label{minibe}
{\rm
Consider the functionals $\cE_\e(\mu):=\min_\beta \F_\e(\mu,\beta)$, representing the elastic energy
induced by the dislocation measure $\mu$.
By the $\Gamma$-convergence result stated in Theorem \ref{mainthm} we immediately deduce that the
functionals $\cE_\e$  $\Gamma$-converge (as $\e\to 0$) to the functional $\cE: \ms(\om;\R^2)\cap H^{-1}(\om;\R^2)\to \R$
defined by
\begin{equation}\label{edimu2}
\cE(\mu):= \min_\beta \{\F(\mu, \beta): \Curl \beta = \mu \}.
\end{equation}
Therefore, the energy $\cE(\mu)$ induced by a distribution of
dislocations $\mu$ in the critical regime is given by the sum of
both an elastic  and a plastic term. In particular,  any
distribution of dislocations in this regime induces a  residual
elastic distortion (i.e., if $\mu$ is not zero, then so is the
corresponding strain $\beta$ that minimizes \eqref{edimu2} and hence
its elastic energy).}
\end{remark}

\subsection{Compactness}
Let $\{(\mu_n,\beta_n)\}$ be a sequence in $\ms(\om;\R^2)\times L^2(\Om;\M^{2\times 2})$ such that
$\F_{\en}(\mu_n,\beta_n)\le  E$ for some positive constant $ E$ independent of $n$.
We give the proof of the compactness property stated in Theorem \ref{mainthm} in three steps.
\vskip5pt
\noindent

{\it Step $1.$ Weak compactness of the rescaled
dislocation measures.}
\vskip3pt
We first show that the sequence $\{(1/{|\log \en|})\mu_n\} $ is uniformly bounded in mass. Let
$\mu_n=\sum_{i=1}^{M_n}  \xi_{i,n} \delta_{x_{i,n}} $, with $\xi_{i,n}\in \mathbb{S}$; we claim that
\begin{equation}\label{liminma}
\frac{1}{|\log \en|} |\mu_n|(\Om)= \frac{1}{|\log \en|} \sum_{i=1}^{M_n} |\xi_{i,n}| \le C.
\end{equation}
Indeed,
$$
E \ge \F_{\en}(\mu_n,\beta_n) = \frac{1}{|\log \en|^2} \int_{\Om_{\e}(\mu_n)} W(\beta_n) \,dx\ge \sum_{i=1}^{M_n}
\frac{1}{|\log \en|^2} \int_{B_{\rho_{\e_n}}(x_{i,n})} W(\beta_n) \, dx\,,
$$
where we recall that $\beta_n=0$ in $\Omega\setminus\Om_\e(\mu_n)$.
After a change of variables we deduce
$$
E \ge \sum_{i=1}^{M_n} \frac{1}{|\log \en|^2} \int_{B_{\rho_{\en}}} W(\beta_n(x_{i,n}+ y)) \, dy.
$$
Note that the functions $y\to \beta_n(x_{i,n}+ y)$ belong to the class
$\asm_{\en,\rho_{\en}}(\xi_{i,n})$ defined in \eqref{adcell}. Therefore we have
\begin{multline}\label{linu}
E \ge \sum_{i=1}^{M_n} \frac{1}{|\log \en|^2} \int_{B_{\rho_{\en}}} W(\beta_n(x_i+ y)) \, dy \\
\ge \frac{1}{|\log \en|} \sum_{i=1}^{M_n} \bar \psi_{\en}(\xi_{i,n}) =
\frac{1}{|\log \en|} \sum_{i=1}^{M_n} |\xi_{i,n}|^2 \, \bar \psi_{\en}\left(\frac{\xi_{i,n}}{|\xi_{i,n}|}\right),
\end{multline}
where the function $\bar\psi_{\en}$ is defined in \eqref{rhopsie0}.
Let $\psi$ be the function given by Corollary \ref{coro},
and let $2c:=\inf_{|\xi|=1} \psi(\xi)$.
By Proposition \ref{equicell} we deduce that for $n$ large enough
$\bar \psi_{\en}(\xi)\ge c$ for every $\xi$ with $|\xi|=1$.
By \eqref{linu} we obtain
$$
E \ge \frac{1}{|\log \en|} \sum_{i=1}^{M_n} |\xi_{i,n}|^2 \bar \psi_{\en}\left(\frac{\xi_{i,n}}{|\xi_{i,n}|}\right) \ge
\frac{c}{|\log \en|} \sum_{i=1}^{M_n} |\xi_{i,n}|^2 \ge  \frac{C}{|\log \en|} \sum_{i=1}^{M_n} |\xi_{i,n}|,
$$
where the last inequality follows from the fact that $\xi_{i,n}\in \mathbb{S}={\rm Span}_\Z S$, $S$ is a finite set,
and hence $|\xi_{i,n}|$ are bounded away from zero. We conclude that  \eqref{liminma} holds.
\par
%Finally we pass to the proof of the compatibility condition satisfied by the pair
%$(\mu,\beta)$.
\vskip8pt
\noindent
{\it Step $2.$ Weak compactness of the rescaled
strains.}
\vskip3pt
In view of the coercivity condition \eqref{coercon} we have
\begin{equation}\label{facile}
C |\log \en|^2 \ge C |\log \en|^2 \F_{\e} (\mu_n,\beta_n) \ge C
\int_{\Om} W(\beta_n) \,dx\ge \int_{\Om} |\bs_n|^2 \, dx.
\end{equation}
The idea of the proof is to apply the generalized Korn inequality provided by Theorem \ref{KTI}
to $\beta_n$ to control the anti-symmetric part of $\beta_n$. Note that the $\Curl$ of $\beta_n$
is clearly related to the dislocation measure $\mu_n$, whose mass is bounded by $C|\log\en|$ by Step 1.
On the other hand, it is not clear that
$|\Curl \beta_n|\le C |\log \en|$. Therefore we proceed as follows.
For every $x_{i,n}$ in the support set of $\mu_n$, set $C_{i,n}:= B_{2 \en}(x_{i,n}) \setminus B_{ \en}(x_{i,n})$
and consider the function
$ K_{i,n}: \, C_{i,n} \to \Md$ defined by
$$
K_{i,n}:=  \frac{1}{2\pi}   \xi_{i,n} \otimes J \frac{x -x_{i} }{|x-x_{i}|^2},
$$
where $J$ is the clockwise rotation of $90^{\text o}$.
It easy to show that the following estimate holds true
\begin{equation}\label{ineq-homogeneity}
\int_{C_{i,n}} |K_{i,n}|^2 \, dx  \le
 c \int_{{C_{i,n}}}|\beta_n\sym|^2 \, dx,
\end{equation}
Indeed, it is straightforward to check that
$$
\int_{C_{i,n}} |K_{i,n}|^2 \, dx = C |\xi_{i,n}|^2
$$
and  by a scaling argument, from Remark~\ref{min-prob} we get
$$
\int_{C_{i,n}} |\beta_n^{\sym}|^2 \, dx \geq c |\xi_{i,n}|^2
$$
and hence (\ref{ineq-homogeneity}).

By construction $\Curl (\beta_n-K_{i,n})=0$ in $C_{i,n}$, and so $\beta_n-K_{i,n} = \nabla v_{i,n}$ in $C_{i,n}$
for some $v_{i,n}\in H^1(C_{i,n};\R^2)$. Thus
by (\ref{ineq-homogeneity}) we have
$$
\int_{C_{i,n}} |\nabla v_{i,n}\sym|^2 \, dx \le C
\int_{{C_{i,n}}}(|\beta_n\sym|^2 + |K_{i,n}|^2 )\, dx \le C
\int_{{C_{i,n}}}|\beta_n\sym|^2 \, dx\,,
$$
and hence, applying the standard Korn's inequality to $v_{i,n}$, we deduce that
$$
\int_{C_{i,n}} |\nabla v_{i,n} - A_{i,n}|^2 \, dx \le C \int_{C_{i,n}}
|\nabla v_{i,n}\sym|^2 \, dx \le C \int_{{C_{i,n}}}|\beta_n\sym|^2 \, dx,
$$
where $A_{i,n}$ is a suitable anti-symmetric matrix. By standard extension arguments there exists
a function $u_{i,n}\in H^1 (B_{2\en}(x_{i,n});\R^2)$, such that $ \nabla u_{i,n} \equiv  \nabla v_{i,n} - A_{i,n}$ in
$C_{i,n}$, and such that
\begin{equation}\label{nuin}
\int_{ B_{2\en}(x_{i,n})} |\nabla u_{i,n}|^2 \, dx  \le \int_{{C_{i,n}}} |\nabla v_{i,n} - A_{i,n}|^2 \, dx \le C \int_{{C_{i,n}}}
|\beta_n\sym|^2 \, dx.
\end{equation}

Consider the field $\tilde \beta_n: \Om\to \Md$ defined by
$$
\tilde \beta_n(x):=
\begin{cases}
\beta_n(x) & \text{ if } x\in \Om_{\en},\\
\nabla u_{i,n}(x) + A_{i,n} & \text{ if } x\in B_{\en}(x_{i,n}).
\end{cases}
$$
In view of \eqref{nuin} it follows
$$
\int_{\Om} |\tilde \beta_n\sym|^2 \, dx = \int_{\Om} | \beta_n\sym|^2 \, dx + \sum_i
\int_{B_{\en}(x_{i,n})} |\nabla u_{i,n}\sym|^2 \, dx \le C \int_{\Om} |\beta_n\sym|^2 \, dx \le C|\log \en|^2.
$$
By construction we have
$
|\Curl \tilde \beta_n|(\Om) = |\mu_n|(\Om);
$
therefore, we can apply Theorem  \ref{KTI} to
$\tilde \beta_n$,  obtaining
\begin{equation}\label{ttt1}
\int_{\Om} |\tilde \beta_n - \tilde A_n|^2 \, dx \le
C \left(\int_{\Om} |\tilde \beta_n\sym|^2 \, dx + \big(|\mu_n|(\Om)\big)^2 \right)
\le C|\log \en|^2,
\end{equation}
where $\tilde A_n$ is the average of the anti-symmetric part of $\tilde \beta_n$.
Since the average of $\beta_n$ is a symmetric matrix, we have
\begin{equation}\label{ttt2}
\int_{\Om_{\en}(\mu_n)} |\beta_n|^2 \, dx =
\int_{\Om_{\en}(\mu_n)} |\beta_n - \tilde A_n|^2  - |\tilde A_n|^2 \, dx
\le
\int_{\Om_{\en}(\mu_n)} |\beta_n - \tilde A_n|^2 \, dx.
\end{equation}
Finally, by \eqref{ttt1} and \eqref{ttt2} we conclude
$$
\int_{\Om_{\en}(\mu_n)} |\beta_n|^2 \, dx \le
\int_{\Om_{\en}(\mu_n)} |\beta_n - \tilde A_n|^2 \, dx  \le \int_{\Om} |\tilde \beta_n - \tilde A_n|^2 \, dx \le C|\log \en|^2,
$$
which  gives the desired compactness property for $\beta_n/|\log\en|$ in $L^2(\Om;\Md)$.
\vskip8pt
\noindent

{\it Step $3.$ $\mu$ belongs to $H^{-1}(\om;\R^2)$ and $\Curl \beta=\mu$.}
\vskip3pt
Let $\f\in C^1_0(\Om)$. It is easy to construct a sequence $\{\f_n\}\subset H^1_0(\Om)$ converging  to  $\f$
uniformly and strongly in
$H^1_0(\Om)$ and satisfying the property
$$
\f_n \equiv \f(x_{i,n}) \qquad \text{ in } B_{\en}(x_{i,n})\quad \text{for every } x_{i,n} \text{ in the support set of } \mu_n.
$$
By Remark \ref{curlbeta} we have
\begin{eqnarray*}
\int_{\Om} \f \, d\mu&=& \lim_{n\to +\infty} \frac{1}{|\log \en|} \int_\Om \f_n \, d\mu_n\
=\ \lim_{n\to +\infty} \frac{1}{|\log \en|} <\Curl \beta_n, \f_n>\\
&=&\lim_{n\to +\infty} \frac{1}{|\log \en|} \int_\om \beta_n J\nabla\f_n\,dx \ =\ \int_\om \beta J\nabla\f\,dx \
=\ <\Curl \beta,\f>,
\end{eqnarray*}
from which we deduce the admissibility condition $\Curl \beta=\mu$. Moreover, since by the previous step
we have $\beta\in L^2(\Om;\Md)$, we deduce that
$\mu$ belongs to $H^{-1}(\om;\R^2)$.

\subsection{$\Gamma$-liminf inequality}
Here we prove the $\Gamma$-liminf inequality of Theorem \ref{mainthm}.
Let
$$
(\mu,\beta) \in \ms(\om;\R^2) \times L^2(\Om;\M^{2\times 2})
\text{ with } \Curl \beta = \mu,
$$
and let $(\mu_\e,\beta_\e)$ satisfy \eqref{not1} and \eqref{not2}.
In order to prove the $\Gamma$-liminf inequality
\begin{equation}\label{glii}
\liminf_{\e\to 0}\F_{\e}(\mu_\e,\beta_\e) \ge \F(\mu,\beta),
\end{equation}
it is enough to show that inequality holds  for the self and the interaction energy separately.
More precisely, we write the energy corresponding to $(\mu_\e,\beta_\e)$  in the following way
\begin{equation*}
\F_{\e}(\mu_\e,\beta_\e)=
\int_\Om \chi_{\cup B_{\rho_\e}(x_{i,\e})} W(\beta_\e) \, dx + \int_{\Om} \chi_{\Om\setminus \cup
B_{\rho_\e}(x_{i,\e})} W(\beta_\e) \, dx
\end{equation*}
and we estimate the two terms separately.
Set  $\eta_\e: =\sum_{i=1}^{M_\e}  \delta_{x_{i,\e}}$.
By Proposition~\ref{equicell} we have

\begin{equation}\label{prieq}
\frac{1}{|\log \e|}\int_\Om \chi_{\cup B_{\rho_\e}(x_{i,\e})} W(\beta_n) \, dx \ge
\int_{\Om} \bar\psi_\e \left(\frac{d\mu_\e}{d \eta_\e}\right) d\eta_\e \ge
\int_{\Om}  \psi \left(\frac{d\mu_\e}{d \eta_\e}\right) (1+ o(\e)) \, d\eta_\e \,,
\end{equation}
where $o(\e)\to 0$ as $\e\to 0$. Since ${\rm Span}_\R \mathbb
S=\R^2$, the convex $1$-homogeneous function $\f$ defined in
\eqref{defi} is finite in $\R^2$, and so continuous. Thus, in view
of Reshetnyak Theorem (see \cite[Theorem 1.2]{R}), we deduce

\begin{multline}\label{primoter}
\liminf_{\e\to 0}\frac{1}{|\log \e|}\int_{\Om}  \psi \left(\frac{d\mu_\e}{d \eta_\e}\right)  d\eta_\e \ge
\liminf_{\e\to 0}\frac{1}{|\log \e|} \int_{\Om} \f\left(\frac{d\mu_\e}{d \eta_\e}\right) \, d \eta_\e
\\
=
\liminf_{\e\to 0}\frac{1}{|\log \e|} \int_{\Om} \f\left(\frac{d\mu_\e}{d|\mu_\e|}\right) \, d|\mu_\e|
\ge
\int_{\Om} \f\left(\frac{d\mu}{d|\mu|}\right) \, d|\mu|.
\end{multline}

From \eqref{prieq} and \eqref{primoter} we have that the $\Gamma$-liminf inequality holds  for the self-energy, i.e.,
\begin{equation}\label{seee}
\liminf_{\e\to 0}\frac{1}{|\log \e|^2}\int_\Om \chi_{\cup B_{\rho_\e}(x_{i,\e})} W(\beta_\e) \, dx
\ge \int_{\Om} \f\left(\frac{d\mu}{d|\mu|}\right) \, d|\mu|.
\end{equation}

Concerning the interaction energy, by semicontinuity we immediately deduce

\begin{multline}\label{interine}
\liminf_{\e\to 0}\frac{1}{|\log \e|^2} \int_{\Om} \chi_{\Om\setminus \cup
B_{\rho_\e}(x_{i,\e})} W(\beta_\e) \, dx \\= \liminf_{\e\to 0}\int_{\Om} W\left(\chi_{\Om\setminus \cup
B_{\rho_\e}(x_{i,\e})} \frac{1}{|\log \e|} \beta_\e\sym\right) \, dx \ge \int_{\Om} W(\beta\sym) \, dx.
\end{multline}

In the last  inequality we used the fact  that the number of atoms of $\mu_\e$ is bounded by $|\log\e|$
and hence, since by assumption $|\log\e|\rho_\e^2\to 0$, the function $\chi_{\Om\setminus \cup
B_{\rho_\e}(x_{i,\e})}$ converges strongly to $1$.

Inequality \eqref{seee}, together with inequality \eqref{interine},
gives the $\Gamma$-liminf inequality.

\subsection{$\Gamma$-limsup inequality}
Here we prove the $\Gamma$-limsup inequality of Theorem \ref{mainthm}.
%Let $\mu\in H^{-1}(\Om;\R^2)\cap \ms(\om;\R^2)$ and let $\beta\in L^2
%(\Om;\M^{2\times 2})$ with $\Curl \beta=\mu$. We want to construct a
%sequence $\{(\mu_\e,\beta_\e)\}$  in $\ms(\om;\R^2)\times L^2(\Om;\M^{2
%\times 2})$ such that
%\begin{equation}\label{gls}
%\limsup_{\e\to 0}\F_{\e}(\mu_\e,\beta_\e) \le \F(\mu,\beta).
%\end{equation}
We begin with a lemma that will be useful in the construction of recovery sequences also for different
energy regimes under consideration.
Given $\mu_\e:= \sum_{i=1}^{M_\e} \xi_{i,\e} \delta_{x_{i,\e}} \in X_\e$ and $r_\e\to 0$, we introduce
the corresponding measures, diffused on
balls of radius $r_\e$ and on circles of radius $r_\e$, respectively, defined by

\begin{equation}\label{tmr}
\tilde\mu_\e^{r_\e}:=\sum_{i=1}^{M_\e}
\frac{\chi_{B_{r_\e}(x_{i,\e})}}{\pi r_\e^2}\xi_{i,\e}, \qquad
\hat\mu^{r_\e}_\e:=\sum_{i=1}^{M_\e} \frac{\Hh^1\res \partial
B_{r_\e}(x_{i,\e})}{2\pi r_\e}\xi_{i,\e}\,.
\end{equation}

For every $x_{\e,i}$ in the support set of $\mu_\e$ we define the functions
$\tilde K_{\e,i}^{ \xi_{\e,i}}, \, \hat K_{\e,i}^{ \xi_{\e,i}}: \, B_{r_\e}(x_{\e,i}) \to \Md$ as follows
\begin{equation}\label{defditbhb}
\tilde K_{\e,i}^{ \xi_{\e,i}}(x):= \frac{1}{2 \pi r_\e^2} \xi_{\e,i} \otimes J (x-x_{\e,i}),
\qquad
\hat K_{\e,i}^ {\xi_{\e,i}}(x):=  \frac{1}{2\pi}  \xi_{\e,i} \otimes J \frac{x -x_{\e,i} }{|x-x_{\e,i}|^2},
\end{equation}
where $J$ is the clockwise rotation of $90^{\text o}$.
Finally, we introduce the functions $\tilde K_\e^{\mu_\e}:\Om\to \R^2$, $\hat K_\e^{\mu_\e}:\Om\to \R^2$,
defined by
\begin{equation}\label{dekeb}
\tilde K_\e^{\mu_\e}:= \sum_{i=1}^{M_\e} \tilde K_{\e,i}^ {\xi_{\e,i}} \chi_{B_{r_\e}(x_{\e,i})}, \qquad
\hat K_\e^{\mu_\e}:= \sum_{i=1}^{M_\e} \hat K_{\e,i}^ {\xi_{\e,i}} \chi_{B_{r_\e}(x_{\e,i})}\,.
\end{equation}
Note that
\begin{equation}\label{rotori}
\Curl \tilde K_\e^{\mu_\e} =\tilde \mu^{r_\e}_\e - \hat \mu^{r_\e}_\e, \qquad
\Curl \hat K_\e^{\mu_\e} = - \hat \mu^{r_\e}_\e.
\end{equation}

\begin{lemma}\label{recovery}
Let $N_\e\to \infty$ be satisfying \eqref{ipn},  $\xi:= \sum_{k=1}^{M} \lambda_k \xi_k$ with
$\xi_k\in \mathbb S$, $\lambda_k\geq 0$, $ \Lambda:= \sum_k \lambda_k$,   $\mu:= \xi \, dx$ and $r_\e:=1/(2 \sqrt{\Lambda N_\e})$.
Then there exists a sequence of measures $\mu_\e = \sum_{k=1}^{M} \xi_k \mu_\e^k$ in $X_\e$, with
$\mu_\e^k$ of the type $\sum_{l=1}^{M_\e^k}  \delta_{x_{\e,l}}$  such that $B_{{r_\e}}(x) \subset\Om, \,
|x -y| \ge 2r_\e$  for every
$x, \, y$ in the support set of $\mu_\e$, and such that
\begin{equation}\label{l1}
|\mu^k_\e|/N_\e \weakst \lambda_k \, dx \text{ in } \ms(\Om),
\end{equation}
\begin{equation}\label{l2}
\frac{\mu_\e}{N_\e} \weakst \mu \text{ in } \ms(\Om;\R^2), \quad \frac{\hat \mu^{r_\e}_\e}{N_\e} \weakst \mu \text{ in } \ms(\Om;\R^2),
\quad \frac{\tilde \mu^{r_\e}_\e}{N_\e} \weakst \xi \text{ in } L^\infty(\Om;\R^2),
\end{equation}
\begin{equation}\label{l3}
\frac{\tilde \mu^{r_\e}_\e}{N_{\e}} \to \mu, \quad
\frac{\tilde \mu^{r_\e}_\e - \hat \mu^{r_\e}_\e}{N_{\e}}
\to 0 \qquad \text{ strongly in } H^{-1}(\Om; \R^2).
\end{equation}
\end{lemma}

\begin{proof}
First, we  prove  the lemma for $M=1$ and $\mu= \xi \, dx$ with $\xi\in \mathbb S$.
\par
For $\mu:= \xi\,dx$ we cover $\R^2$ with cubes  of size $2r_\e$, we
plug a mass with weight $\xi$ in the center of all of such  cubes
which are contained in $\Om$, and we set $\mu_\e$ the measure
obtained through this procedure. Let us prove \eqref{l3}, all the
other properties following easily by the definition of $\mu_\e$.
Since $\mu - \frac{\tilde \mu^{r_\e}_\e}{N_{\e}}$ converges weakly
to zero in $L^2(\Om;\R^2)$, by the compact embedding of $L^2$ in
$H^{-1}$ we have that $\mu - \frac{\tilde\mu^{r_\e}_\e}{N_{\e}} \to
0$ in $H^{-1}(\Om;\R^2)$. The fact   that $\frac{\tilde
\mu^{r_\e}_\e - \hat \mu^{r_\e}_\e}{N_{\e}} \to 0 \text{ strongly in
} H^{-1}(\Om; \R^2)$ follows directly by \eqref{rotori}, since
$\tilde K_\e^{\mu_\e}/ N_\e \to 0$ strongly in $L^2(\Om;\R^2)$ (that
can be checked by a very simple estimate).

The general case with $\xi\in\R^2$ and $M>1$ will follow by approximating $\mu=\xi dx$ with
periodic locally constant measures with weight $\xi_k$ on sets with volume fraction $\lambda_k/\Lambda$.
In each region where the approximating measure is constant we apply the construction above and then we take a diagonal sequence.
\end{proof}

We are in a position to prove the $\Gamma$-limsup inequality of Theorem \ref{mainthm}.
%begin with the construction of the recovery sequence.
We will proceed in several steps.
\vskip8pt
\noindent
{\it Step $1.$ The case $\mu\equiv \xi \,dx$.}
\vskip3pt
Given $\xi\in\R^2$ and $\beta\in L^2(\Om;\Md)$ with $\Curl \beta=\xi\, dx$,  we will construct
a recovery sequence $\{\mu_\e\} \subset X_{\e}$,  $\beta_\e\in \asmn(\mu_\e)$, such that $(\mu_\e,\beta_\e)$
converges to $(\xi \, dx, \beta)$ in the sense of (\ref{not1}) and (\ref{not2}),
\begin{equation}\label{glsa}
\limsup_{\e\to 0} \frac{1}{|\log \e|^2} \int_\om W(\beta_\e) \, dx
\le \int_\om (W(\beta)  + \f(\xi))\, dx\,,
\end{equation}
and  satisfying the additional requirement that $(\beta_\e/|\log \e| - \beta) \cdot t$ tends to zero strongly
in $H^{-\frac{1}{2}}(\partial \om)$.

Let $\f$ be the self-energy density defined in \eqref{defi}, and let $\lambda_k\geq 0$, $\xi_k\in \mathbb S$
be such that $\xi=\sum_k\lambda_k\xi_k$ and
\begin{equation}\label{recf}
\f(\xi)= \sum_{k=1}^M \lambda_k \psi(\xi_k).
\end{equation}
Consider the sequence $\mu_\e:= \sum_{i=1}^{ M_\e} \xi_{\e,i} \delta_{x_{\e,i}}$ given by Lemma \ref{recovery}
with $N_\e=|\log \e|$. Note that, since $N_\e\rho_\e^2\to 0$, we have that $r_\e>\!\!>\rho_\e$, and so $\mu_\e\in X_\e$.

Since the function $\hat K^{\xi_{\e,i}}_{\e,i}$ defined in (\ref{defditbhb}) belongs to $ \asm_{\e,\rho_\e}(\xi_{\e,i})$
and  satisfies condition (\ref{pallino}), by Proposition \ref{equicell} for every $x_{\e,i}$ in the support set of $\mu_\e$
we can find a strain
$\hat \beta_{\e,i}:\Omega \to \Md$ such that
\begin{itemize}
\item[1)] $\hat \beta_{\e,i}\in \asm_{\e,\rho_\e}(\xi_{\e,i})$,
\item[2)] $\hat\beta_{\e,i}\cdot t = \hat K^{\xi_{\e,i}}_{\e,i}\cdot t$ on $\partial B_\e(x_{\e,i})
\cup\partial B_{\rho_\e}(x_{\e,i})$,
\item[3)] $\frac{1}{|\log\e|}\displaystyle{\int_{B_{\rho_\e}(x_{\e,i})\setminus B_\e(x_{\e,i})} W(\hat \beta_{\e,i})\, dx
= \psi(\xi_{\e,i}) (1 + o(\e))}$ where
$o(\e) \to 0 $ as $\e\to 0$.
\end{itemize}
Extend  $\hat \beta_{\e,i}$ to be $\hat K^{\xi_{\e,i}}_{\e,i}$ in $B_{r_\e}(x_{\e,i})\setminus B_{\rho_\e}(x_{\e,i})$
and  zero in $\Omega\setminus(B_{r_\e}(x_{\e,i})\setminus B_\e(x_{\e,i}))$,
and set
\begin{equation}\label{dekeb2}
\hat \beta_\e:= \sum_{i=1}^{M_\e} \hat \beta_{\e,i} . \end{equation}
Then
\begin{itemize}
\item[4)] $\Curl \hat \beta_{\e} =- \hat \mu^{r_\e}_\e+\hat\mu^\e_\e$, \end{itemize}
where $\hat\mu^{r_\e}_\e$ and $\hat\mu^\e_\e$ are defined  in (\ref{tmr}). Finally, set
\begin{equation}\label{ultima}
\bar \beta_\e:= |\log\e| \beta - \tilde K_\e^{\mu_\e} + \hat \beta_\e,
\end{equation}
where $\tilde K_\e^{\mu_\e}$ is defined according to \eqref{dekeb}.
By Lemma \ref{recovery} and  \eqref{rotori}
$$
\Curl \frac{\bar \beta_\e}{|\log \e|}\LL \Omega_\e(\mu_\e)= (\mu -\frac{ \tilde \mu^{r_\e}_\e}{|\log\e|}
+\frac{ \hat \mu^{r_\e}_\e}{|\log\e|}-\frac{ \hat \mu^{r_\e}_\e}{|\log\e|})
=(\mu -\frac{ \tilde \mu^{r_\e}_\e}{|\log\e|}) \to 0 \quad \text{ in }
H^{-1}(\Om; \R^2).
$$
Therefore, we can add to $\bar\beta_\e$ a vanishing sequence $\frac{R_\e}{|\log\e|} \to 0$ in
$L^{2}(\Om;\Md)$, obtaining the admissible strain
\begin{equation}\label{Reps}
\beta_\e:= (\bar \beta_\e + R_\e) \chi_{\Om_\e(\mu_\e)}.
\end{equation}
In order to prove that the pair $(\mu_\e,\beta_\e)$ is the desired recovery sequence
we have to check  the following properties
\begin{itemize}
\item[i)]  $\beta_\e$ converge to $\beta$ in the sense of definition \eqref{not2};
\item[ii)] The pair $(\beta_\e,\mu_\e)$ is a recovery sequence, i.e.,
$$
\lim_{\e\to 0} \int_\om W(\beta_\e) \, dx= \int_\om (W(\beta)  +
\f(\xi)) \,dx .
$$
\end{itemize}
To prove i),  we first note that since $M_\e\sim |\log\e|$, $r_\e\sim 1/\sqrt\e$, we have that
$$
\int_{\om_{\rho_\e}(\mu_\e)}\frac{|\hat \beta_\e|^2}{|\log\e|}dx\,
= \int_{\om_{r_\e}(\mu_\e)\setminus \om_{\rho_\e}(\mu_\e)}\frac{|\hat K^{\mu_\e}_\e|^2}{|\log\e|}dx\,\to 0\,,
$$
which implies that $\hat\beta_\e/|\log\e|$ is concentrated on the hard core region. Then
by Lemma \ref{recovery}, $|\mu^k_\e|/N_\e \weakst \lambda_k \, dx$ for every $k$, and by property 3) we have
\begin{multline}\label{segls}
\lim_{\e\to 0} \frac{1}{|\log \e|^2}\int_{\om} W(\hat \beta_\e) \, dx
= \lim_{\e\to 0} \frac{1}{|\log \e|^2}\int_{\om\setminus\om_{\rho_\e}(\mu_\e)} W(\hat \beta_\e) \, dx\\
=\lim_{\e\to 0} \frac{1}{|\log \e|} \sum_{k=1}^M |\mu_{\e}^k|(\Om) \psi(\xi_k) = \sum_{k=1}^M \lambda_k \psi(\xi_k) = \f(\xi).
\end{multline}
In particular, we deduce that $ \hat \beta_\e/|\log\e|$ is bounded in $L^2(\om;\Md)$.
Since the $L^2$ norm of  $\hat \beta_\e/|\log\e|$ is concentrating on the hard core region,
we conclude that $ \hat \beta_\e/ {|\log\e|}$
converges  weakly  to zero in $L^2(\om;\Md)$. On the other hand, one
can  check directly that $\tilde K_\e^{\mu_\e} /|\log\e|$ converges
strongly to zero in $L^2(\Omega)$. Recalling that also
$R_\e/|\log\e|\to 0$, by \eqref{ultima} and \eqref{Reps} we conclude
that $i)$ holds.
\par
Next we prove ii), namely,  that   the pair $(\beta_\e,\mu_\e)$ is optimal in energy. We have
$$
\lim_{\e\to 0} \frac{1}{|\log \e|^2}\int_\om W(\beta_\e) \, dx= \lim_{\e\to 0} \frac{1}{|\log \e|^2}
\int_\om W(|\log \e| \beta + \hat \beta_\e)\, dx.
$$
Since $\hat \beta_\e/|\log \e| \weak 0$ in $L^2(\Om;\Md)$, taking into account also \eqref{segls}, we conclude
$$
\lim_{\e\to 0} \frac{1}{|\log \e|^2}\int_\om W(\beta_\e)\,dx= \lim_{\e\to 0} \frac{1}{|\log \e|^2}
\left( \int_{\om} W(|\log \e| \beta)\,dx + \int_{\om} W(\hat\beta_\e)\,dx \right) =
\int_\om (W(\beta) + \f(\xi))\,dx.
$$
Finally, by the Lipschitz continuity of $\partial \om$, from
\eqref{ultima} and (\ref{Reps}) we also deduce that $(\beta_\e/|\log
\e| - \beta) \cdot t$ tends to zero strongly in
$H^{-\frac{1}{2}}(\partial \om)$.

\vskip8pt
\noindent
{\it Step $2.$ The case $\mu := \sum_{l=1}^L \chi_{A_l} \xi_l \,dx$.}
\vskip3pt

In this step we  proof  the $\Gamma$-limsup inequality
in the case of $\mu$ locally constant, i.e., of the type
\begin{equation}\label{loco}
\mu:= \sum_{l=1}^L \chi_{A_l} \xi_l\,dx,
\end{equation}
where $A_l$ are open subsets of $\Om$ with Lipschitz continuous boundary and $\xi_i\in \M^{2\times 2}$.
The construction of the recovery sequence is based on classical localization arguments in $\Gamma$-convergence and
takes advantage of the previous step.
\par
Let us set $\beta_l:= \beta \res A_l$, and let $\mu_{l,\e}$, $\beta_{l,\e}$ be the recovery sequence
given by Step $1$ applied with $\Om$ replaced by $A_l$, with $\beta=\beta_l$, $\mu=\xi_l\,dx$. Finally
let us set $\bar \beta_\e:\Om\to \Md$ and $\mu_\e\in \ms(\om;\R^2)$ as follows
$$
\bar \beta_\e(x):= \beta_{l,\e} \qquad \text{ if } x\in A_l, \qquad \mu_\e:= \sum_l
\mu_{l,\e}.
$$
By construction we have $\mu_\e\in X_\e$. Moreover  since
$$
\frac{1}{|\log \e|}\|\Curl \bar\beta_\e\res \om_\e(\mu_\e) \|_{H^{-1}(\Om; \R^2)}
\leq \sum_l \left\|\frac{\beta_{\e,l}}{|\log\e|}-\beta\right\|_{H^{-\frac{1}{2}}(\partial A_l)}\,,
$$
from the previous step it easily follows that
$$
\frac{1}{|\log \e|}\Curl \bar\beta_\e\res\om_\e(\mu_\e) \to 0
$$
strongly in $H^{-1}(\Om;\R^2)$.
Therefore we can easily modify this sequence $\bar \beta_\e$ by adding a vanishing perturbation
in order to obtain the desired recovery sequence  $\beta_\e$.

\vskip8pt
\noindent
{\it Step $3.$ The general case.}
\vskip3pt

In this step we show how to construct a recovery sequence in the general case
(namely for a general dislocation measure $\mu \in \ms(\Om;\R^2)$).
\par
Let $(\mu,\beta)$ be given  in the domain of the $\Gamma$-limit $\F$. In view of the previous step
and by standard density arguments in $\Gamma$-convergence, it is enough to  construct  sequences
$(\mu_n,\beta_n)$
with $\Curl \beta_n=\mu_n$ and  with $\mu_n$ locally constant as in \eqref{loco}  such that
\begin{equation}\label{aim}
\beta_n \to \beta \text{ in } L^2(\Om;\M^{2\times 2}), \qquad
\mu_n \weakstar \mu \text{ in }\ms(\Om;\M^{2\times2})\quad \text{ and } \quad  |\mu_n|(\om)\to |\mu|(\om)  .
\end{equation}
Indeed,  under these convergence assumptions we get the convergence
of the corresponding energies, i.e.,
\begin{equation}\label{mainaim}
\lim_{n\to +\infty} \F(\beta_n,\mu_n)=\F(\beta,\mu)\,.
\end{equation}
By standard reflection arguments we can extend the strain $\beta$ to a function $\beta_A$
defined in a neighborhood $A$ of $\Om$, such that $\Curl \beta_A = \mu_A$ is a measure on
$A$ and moreover $| \mu_A|(\partial \Om)=0$.
\par
Let $\rho_h$ be a sequence of mollifiers, and define
$$
f_h:= \beta_A * \rho_h \res \Om, \qquad
g_h:= \mu_A * \rho_h \res \Om.
$$
For $h$ large enough these objects are well defined in $\Om$ and $\Curl f_h= g_h$. Clearly
\begin{equation}\label{ar-1}
f_h \to \beta \text{ in } L^2(\Om;\M^{2\times 2}),
\qquad
g_h \, dx \weakstar \mu \text{ in }\ms(\om;\R^2).
\end{equation}
Moreover,
since $|\mu_A|(\partial \Om)=0$, we have
\begin{equation}\label{ar-11}
|g_h \, dx|(\om)\to |\mu|(\om)  .
\end{equation}
Next, we approximate every $g_h$ by locally constant functions, precisely,  we consider
a locally constant function $g_{h,k}$ such that
\begin{equation}\label{ar0}
|g_{h,k} -g_h|_{L^\infty(\Om;\R^2)} \to 0  \text{ as } k\to \infty,\qquad\text{ and } \qquad \int_\om g_{h,k} - g_h\, dx =0\,.
\end{equation}
Let $r_{h,k}$ be the solution of the following problem
\begin{equation}\label{ar1}
\left\{
\begin{array}{ll}
\Curl r_{h,k} = g_{h,k} - g_h & \text{ in } \Om, \\
\text{Div } r_{h,k} = 0 & \text{ in } \Om , \\
r_{h,k}\cdot t = 0 & \text{ in }\partial \Om.
\end{array}
\right.
\end{equation}
By standard elliptic estimates we have
\begin{equation}\label{ar2}
|r_{h,k}|_{L^2(\Om;\M^{2\times 2})} \le C |g_{h,k} - g_h|_{L^2(\Om;\R^2)}.
\end{equation}
Finally, we set $f_{h,k}:= f_h + r_{h,k}$. By \eqref{ar1} we have
$\Curl f_{h,k}= g_{h,k}$. Moreover, by \eqref{ar0}, \eqref{ar2} we have
\begin{equation}\label{ar3}
f_{h,k} \to f_h \text{ in } L^2(\Om;\M^{2\times 2}) \quad \text{ as } k\to\infty.
\end{equation}
By \eqref{ar-1}, \eqref{ar0}, \eqref{ar3}, using a diagonal argument we can find
a sequence $(\mu_n,\beta_n)$ satisfying \eqref{aim}, and therefore \eqref{mainaim}.

\section{The sub-critical case ($N_\e<\!\!< |\log\e|$)}\label{sub}
In this section we study the asymptotic behavior of the energy
functionals $E_\e$ defined in \eqref{elaene}
in the case of dilute dislocations, i.e., for $N_\e <\!\!< |\log\e|$. In terms of $\Gamma$-convergence, it means that we rescale
$E_\e$ with a prefactor $N_\e |\log \e|$, with $N_\e <\!\!< |\log \e|$.
As we discussed in Section \ref{Scalings}, in this case  the  self-energy for minimizing
sequences is predominant with respect to the interaction energy (see Remark \ref{selfpre}).
\par
In contrast  to the critical case, we have that the prefactors of strains and dislocation
measures in the sub-critical case are different. Indeed, the natural rescaling for the
dislocation measures is given by $N_\e$. On the other hand, in order to catch the effect
of the diffuse  energy associated with a sequence $(\mu_\e,\beta_\e)$ with bounded energy
we have to rescale the strains  by $(N_\e |\log\e|)^{1/2}$. These two quantities clearly
coincide only in the critical case $N_\e \equiv |\log\e|$. The effect of a different
rescaling for strains and dislocation measures is that in the limit configuration
$\mu$ and $\beta$ are independent variables, i.e.,  the compatibility condition
$\Curl \beta=\mu$ disappears in this limit. Actually the admissible strains
in the limit are always gradients, i.e., $\Curl \beta=0$. Heuristically this is a consequence of the fact that
the total variation of $\Curl \beta_\e$ is of order $N_\e$, so that
$\Curl (\beta_\e/(N_\e|\log\e|)^{1/2})$ vanishes.
\par
The candidate $\Gamma$-limit of the functionals $\Fdi_\e:
\ms(\om;\R^2) \times L^2(\Om;\Md)$ defined in \eqref{renefa1}  is
the functional  $\F^{\text{dilute}}$ defined by
\begin{equation}\label{glsc}
\Fdi(\mu,\beta):=
\begin{cases}
\displaystyle{\int_{\Om} W(\beta) \, dx + \int_{\Om} \varphi\left(\frac{d\mu}{d|\mu|}\right) \, d|\mu|}
& \text{ if } \Curl\beta = 0;\\
+\infty & \text{ otherwise in } L^2(\Om;\Md).
\end{cases}
\end{equation}
 The precise  $\Gamma$-convergence result is the following.

\begin{theorem}\label{mainthmsc2}
Let $N_\e\to \infty$ be  such that $N_\e/|\log\e|\to 0$. Then the following  $\Gamma$-convergence result holds.
\begin{itemize}
\item[i)]{\bf Compactness.}
Let $\en\to 0$ and let $\{(\mu_n,\beta_n)\}$ be a sequence in
$\ms(\om;\R^2)\times L^2(\Om;\M^{2\times 2})$ such that
$\Fdi_{\en}(\mu_n,\beta_n)\le  E$ for some positive constant $ E$
independent of $n$. Then there exist $\mu\in \ms(\om;\R^2)$ and
$\beta\in L^2(\Om;\M^{2\times 2})$, with $\Curl \beta =0$, such that
(up to a subsequence)
\begin{equation}\label{not1sc2}
\frac{1}{N_{\en}}\mu_n \weakst \mu \qquad {\rm in}\quad \ms(\om;\R^2)\,,
\end{equation}
\begin{equation}\label{not2sc2}
\frac{1}{(N_{\en} |\log\en|)^{1/2}} \beta_n \weak \beta \qquad {\rm in}\quad L^2(\Om;\M^{2\times 2})\,.
\end{equation}
\item[ii)]{\bf $\Gamma$-convergence.}
The functionals  $\F_{\e}$ $\Gamma$-converge as ${\e} \to 0$, with
respect to the convergence in \eqref{not1sc2} and \eqref{not2sc2},
to the functional $\Fdi$ defined  in \eqref{glsc}. More precisely,
the following inequalities hold.
\begin{itemize}
\item[]{$\Gamma$-liminf inequality:}
for every  $(\mu,\beta) \in \ms(\om;\R^2) \times L^2(\Om;\Md)$, with $\Curl \beta=0$,
and for every sequence
$(\mu_\e,\beta_\e) \in X_{\e}\times L^2(\Om;\M^{2\times 2})$
satisfying \eqref{not1sc2} and \eqref{not2sc2}, we have
$$
\liminf_{\e\to 0}\Fdi_{\e}(\mu_\e,\beta_\e) \ge \Fdi(\mu,\beta);
$$
\item[]{$\Gamma$-limsup inequality:}
given $(\mu,\beta)  \in \ms(\om;\R^2) \times L^2(\Om;\Md) $, with $\Curl \beta = 0$, there exists
$(\mu_\e,\beta_\e)\in X_{\e}\times L^2(\Om;\M^{2\times 2})$ satisfying \eqref{not1sc2} and \eqref{not2sc2}  such that
$$
\limsup_{\e\to 0}\Fdi_{\e}(\mu_\e,\beta_\e) \le \Fdi(\mu,\beta).
$$
\end{itemize}
\end{itemize}
\end{theorem}

\begin{remark}\label{selfpre}
{\rm The independence of strains and dislocation measures in the
$\Gamma$-limit is  a consequence of the fact that, in the dilute
regime, the interaction energy is a lower order term  with respect
to the self-energy. Indeed, by the $\Gamma$-convergence result
stated in Theorem \ref{mainthmsc2} we immediately deduce that the
functionals $\cE^{\text{dilute}}_\e$, defined by
\begin{equation*}
\cE^{\text{dilute}}_\e(\mu):= \min_{\beta\in \asmn(\Om)}
\Fdi_\e(\mu,\beta),
\end{equation*}
$\Gamma$-converge (as $\e\to 0$) to the functional $\cE^{\text{dilute}}: \ms(\om;\R^2)\to \R$ defined by
\begin{equation*}
\cE^{\text{dilute}}(\mu):= \int_{\Om}
\varphi\left(\frac{d\mu}{d|\mu|}\right) \, d|\mu|.
\end{equation*}
The energy $\cE^{\text{dilute}}(\mu)$ represents the energy stored
in the crystal induced by the distribution of dislocations $\mu$ in
the dilute regime, and it is given only by the self-energy. }
\end{remark}

The proof of Theorem \ref{mainthmsc2} follows the lines of the proof of Theorem \ref{mainthm}.
For the reader convenience we sketch its main steps.

\begin{proof}[Proof of Theorem \ref{mainthmsc2}]
The compactness property of rescaled strains and dislocation measures stated in
\eqref{not1sc2} and \eqref{not2sc2} can be proved with minor changes as in the
critical case. Let us prove that in this case $\Curl \beta=0$.  Let $\f\in C^1_0(\Om)$ and
let $\{\f_n\}\subset H^1_0(\Om)$ be a sequence converging  to  $\f$ uniformly and strongly in
$H^1_0(\Om)$ and satisfying the property
$$
\f_n \equiv \f(x_{i,n}) \qquad \text{ in } B_{\en}(x_{i,n}) \quad
\text{ for every } x_{i,n} \text{ in the support set of } \mu_n.
$$
By Remark \ref{curlbeta} we have
\begin{multline*}
<\Curl \beta,\f>=\lim_{\en\to 0} \frac{1}{N^{1/2}_{\en} |\log \en|^{1/2}} <\Curl \beta_n, \f_n>  \\ =
\lim_{\en\to 0} \frac{N^{1/2}_{\en}}{|\log \en|^{1/2}} \frac{1}{N_{\en}} \int_\Om \f_n \, d\mu_n = \lim_{\en\to 0} \frac{N^{1/2}_{\en}}{|\log \en|^{1/2}} \int_{\Om} \f \, d\mu=0
\end{multline*}
from which we deduce  $\Curl \beta=0$.

\par
Concerning the $\Gamma$-convergence result, the proof of $\Gamma$-liminf inequality is
identical to that of the critical case, so that we pass directly to the proof of the $\Gamma$-limsup inequality.
\par

As in the critical case, classical localization arguments reduce the problem to the case of  $\mu$  constant.
(Note that the density argument used in the critical case is even easier in the sub-critical case,
since no admissibility condition $\Curl \beta= \mu$ is required.)
\par
The proof of the $\Gamma$-limsup inequality reduces to find  a
sequence $\{\mu_\e\} \subset X_{\e}$, with $1/N_{\e}\mu_\e \weakstar
\xi\,dx$ in $\mathcal{M}(\Om;\R^2)$, and  a sequence $\beta_\e\in
\asmn(\mu_\e)$, with $1/(N_{\en}|\log \en|)^{1/2} \beta_\e  \weak
\beta$ in $L^2(\Om;\Md)$, such that
\begin{equation}\label{glsasc}
\limsup_{\e\to 0}\frac{1}{N_{\en} |\log \en|} \int_\Om
W(\beta_\e)\,dx \le \int_\Om (W(\beta) + \f(\xi)) \,dx,
\end{equation}
and  satisfying the additional requirement that $(\beta_\e/|\log \e| - \beta) \cdot t$ tends to zero strongly in $H^{-\frac{1}{2}}(\partial \om)$.

Consider the sequence $\mu_\e:= \sum_{i=1}^{M_\e} \xi_{\e,i} \delta_{x_{\e,i}}$ given by Lemma \ref{recovery}.
Construct  the functions $\hat \beta_\e:\Om\to \R^2$
as in \eqref{dekeb2}. Then set
\begin{equation}\label{ultima2}
\bar \beta_\e:= (N_\e |\log\e|)^{\frac{1}{2}} \beta - \tilde K_\e^{\mu_\e} + \hat \beta_\e,
\end{equation}
where $\tilde K_\e^{\mu_\e}$ is defined according to \eqref{dekeb}.
By \eqref{rotori}
$$
\Curl \bar \beta_\e \LL \Omega_\e(\mu_\e)= -  \tilde \mu^{r_\e}_\e
$$
By its  definition the density of the measure $\tilde \mu^{r_\e}_\e/{(N_\e|\log \e|)^{\frac{1}{2}}}$
tends to zero uniformly as $\e\to 0$.
Therefore, we can add to $\bar\beta_\e$ a  sequence $R_\e$ in $L^{2}(\Om;\Md)$, with
${R_\e}/{(N_\e|\log \e|)^{\frac{1}{2}}} \to 0$,
obtaining the admissible strain
\begin{equation}\label{Reps2}
\beta_\e:= (\bar \beta_\e + R_\e) \chi_{\Om_\e(\mu_\e)}.
\end{equation}
In order to prove that the pair $(\mu_\e,\beta_\e)$ is the desired recovery sequence,
we have to check  the following properties
\begin{itemize}
\item[i)]  $\beta_\e$ converge to $\beta$ in the sense of definition \eqref{not2sc2};
\item[ii)] The pair $(\beta_\e,\mu_\e)$ is a recovery sequence, i.e.,
$$
\lim_{\e\to 0} \int_\om W(\beta_\e) \, dx= \int_\om (W(\beta)  +
\f(\xi)) \,dx .
$$
\end{itemize}
To prove i),
recalling that by Lemma \ref{recovery}, $|\mu^k_\e|/N_\e \weakst \lambda_k \, dx$ for every $k$, we have
\begin{equation}\label{seglssc}
\lim_{\e\to 0} \frac{1}{N_\e |\log\e|}\int_{\om} W(\hat \beta_\e) \, dx
= \lim_{\e\to 0} \frac{1}{N_\e} \sum_{k=1}^M |\mu_{\e}^k|(\Om) \psi(\xi_k)
= \sum_{k=1}^M \lambda_k \psi(\xi_k) = \f(\xi).
\end{equation}
We deduce that $ \hat \beta_\e/(N_\e |\log\e|)^{1/2}$ is bounded in $L^2(\om;\Md)$. As in the critical case,
the $L^2$ norm of  $\hat \beta_\e/|\log\e|$ is concentrating on the hard core region, so that   $ \hat \beta_\e/ {|\log\e|}$
converges  weakly  to zero in $L^2(\om;\Md)$. On the other hand one can  check directly
that $\tilde K_\e^{\mu_\e} / (N_\e |\log\e|)^{1/2}$ converges strongly to zero in $L^2(\Omega)$, from which  property i) follows.
\par
Concerning ii),  we prove that   the pair $(\beta_\e,\mu_\e)$ is optimal in energy. We have
$$
\lim_{\e\to 0} \frac{1}{(N_\e |\log \e|)}\int_\om W(\beta_\e) \, dx= \lim_{\e\to 0} \frac{1}{(N_\e |\log \e|)}
\int_\om W(|\log \e| \beta + \hat \beta_\e)\, dx.
$$
Since $ \hat \beta_\e/ {|\log\e|} \weak 0$ in $L^2(\Om;\Md)$, taking into account also \eqref{seglssc}, we conclude
\begin{multline*}
\lim_{\e\to 0} \frac{1}{N_\e |\log \e|}\int_\om W(\beta_\e)\,dx
= \lim_{\e\to 0} \frac{1}{N_\e |\log \e|}
\left( \int_{\om} W((N_\e |\log \e|)^{\frac{1}{2}} \beta)\,dx + \int_{\om} W(\hat\beta_\e)\,dx \right)
\\ =
\int_\om (W(\beta) + \f(\xi))\,dx.
\end{multline*}
Finally, by the Lipschitz continuity of $\partial \om$, from \eqref{ultima2} and (\ref{Reps2})
we also deduce that $(\beta_\e/(N_\e|\log \e|)^{1/2} - \beta) \cdot t$ tends to zero strongly in $H^{-\frac{1}{2}}(\partial \om)$.
\end{proof}

\begin{remark}
{\rm
The case $N_\e\le C$ has been considered in \cite{CL}, where  the asymptotic behavior of
the elastic energy for a fixed distribution of dislocations $\mu_\e\equiv \mu$ is provided up
to the second order, and in \cite{Po}, where  the problem of $\Gamma$-convergence induced by  screw dislocations  is addressed
(with $E_\e(\beta):= \|\beta\|_2^2$), without any assumption involving  the notion of hard core region
(essentially with $\rho_\e \approx \e$).
\par
We could extend the result given by Theorem \ref{mainthmsc2} to the case $N_\e \le C$, obtaining
a $\Gamma$-limit which has still the form as in \eqref{glsc},
but with $\mu:= \sum_{i=1}^M \xi_i \, \delta_{x_i}$, where $\xi_i\in \mathbb S$, and with $\f:\mathbb S \to \R$ defined now by
$$
\varphi(\xi):= \inf \left \{ \sum_{k=1}^N  \psi(\xi_k): \, \sum_{k=1}^N  \xi_k = \xi, \, N\in\N,  \, \xi_k\in \mathbb S \right \}.
$$
}
\end{remark}

\section{The super-critical case ($N_\e>\!\!> |\log\e|$)}\label{sup}
In this section we study the asymptotic behavior of the energy
functionals $E_\e$ defined in \eqref{elaene}
in the super-critical case, i.e., for $N_\e >\!\!> |\log\e|$. In terms of $\Gamma$-convergence, it means that we rescale
$E_\e$ by $N_\e^2$, obtaining the rescaled energy functionals $\Fsu_\e$ defined in \eqref{renefa3}.
As we discussed in Section \ref{Scalings}, in this case we  have that the  interaction energy for
minimizing sequences is predominant with respect to the self-energy.
\par
The natural rescaling for the strains in this case is given by $N_\e$, but we don't have any control
on the total variation of the dislocation measure. As a consequence we will get a limit energy
$\Fsu$ defined on  $L^2(\Om;\Mdsym)$ depending on strains   given by
\begin{equation}\label{glsu}
\Fsu(\beta\sym):=
\int_{\Om} W(\beta\sym) \, dx \quad
\text{ if } \beta\sym\in L^2(\om;\Mdsym),
\end{equation}
where $\Mdsym$ denotes the class of symmetric matrices $M\sym$ in $\Md$.
 The precise  $\Gamma$-convergence result is the following.

\begin{theorem}\label{mainthm3}
Let $N_\e$ be such that $N_\e/|\log \e|\to \infty$ as $\e\to 0$. Then the following $\Gamma$-convergence result holds.
\begin{itemize}
\item[i)]{\bf Compactness.}
Let $\en\to 0$ and let $\{(\mu_n,\beta_n)\}$ be a sequence in $\ms(\om;\R^2)\times L^2(\Om;\M^{2\times 2})$
such that $\Fsu_{\en}(\mu_n,\beta_n)\le  E$ for some positive constant $ E$ independent of $n$.
Then there exists a  strain $\beta\sym\in L^2(\Om;\Md)$, such that (up to a subsequence)
\begin{equation}\label{not2su}
\frac{1}{N_{\en}} \beta_n\sym \weak \beta\sym \quad \text{ in } L^2(\Om;\Mdsym).
\end{equation}
%In particular ${1}/{N_{\en}}\mu_n \weak \mu=\Curl\beta$ in $H^{-1}(\Om; \R^2)$

\item[ii)]{\bf $\Gamma$-convergence.}
The functionals  $\Fsu_{\e}$ $\Gamma$-converge as ${\e} \to 0$, with respect to
the convergence in  \eqref{not2su}, to the functional $\Fsu$ defined  in \eqref{glsu}.
More precisely, the following inequalities hold.
\begin{itemize}
\item[]{$\Gamma$-liminf inequality:}
for every  $\beta\sym \in L^2(\Om;\Mdsym)$
and for every sequence
$(\mu_\e,\beta_\e) \in X_{\e}\times L^2(\Om;\M^{2\times 2})$
satisfying  \eqref{not2su} we have
$$
\liminf_{\e\to 0}\Fsu_{\e}(\mu_\e,\beta_\e) \ge \Fsu(\beta\sym);
$$
\item[]{$\Gamma$-limsup inequality:}
given  $\beta\sym\in L^2(\Om;\Mdsym)$ there exists
$(\mu_\e,\beta_\e)\in X_\e\times L^2(\Om;\M^{2\times 2})$ satisfying \eqref{not2su}  such that
$$
\limsup_{\e\to 0}\Fsu_{\e}(\mu_\e,\beta_\e) \le \Fsu(\beta\sym).
$$
\end{itemize}
\end{itemize}
\end{theorem}

%\begin{remark}
%{
%\rm
%Let $(\mu_{\en}, \beta_{\en})$ be a sequence with bounded energy. By the
%compactness property stated in Theorem \ref{mainthm3} we deduce that $
%\beta_{\en}/N_{\en} \weak \beta$ for some limit strain $\beta$. This implies
%$\Curl \beta_{\en} \weakstar \Curl \beta$  in $H^{-1}(\Om; \R^2)$. Moreover,
%as observed in Remark \ref{regmis}, we have that also the diffuse measures $
%\tilde \mu_{\en}$ and $\hat \mu_{\en}$ defined in \eqref{diffusa} and \eqref
%{sulbordo}
%respectively  converge  weakly star to  $\Curl \beta$ in $H^{-1}(\Om; \R^2)$.
%Therefore,  by the $\Gamma$-convergence result stated in Theorem \ref
%{mainthm3} we immediately deduce that the functionals
%$\cE\sce: X_\e \to \R$ defined by
%$$
%\cE\sce(\mu):= \min_{\beta\in\asmn (\mu)} \int_{\Om} W(\beta) \, dx
%$$
%$\Gamma$-converge (as $\e\to 0$) to the functional $\cE^{super}: H^{-1}
%(\om;\R^2)\to \R$ defined by
%\begin{equation}\label{edimudi2}
%\cE^{super}(\mu):= \min \left\{ \int_{\Om} W(\beta(x)) \, dx, \Curl \beta=
%\mu \right\}
%\end{equation}
%with respect to the weak star convergence in $H^{-1}(\Om;\R^2)$ of the
%corresponding diffuse measures.
%The energy $\cE^{super}(\mu)$ represents the energy stored in the crystal
%induced by the distribution of dislocations $\mu$ in the super critical regime, and
%it is given only by the interaction energy. We stress that $\mu$ is just an
%element of $H^{-1}(\Om;\R^2)$, and it is not necessarily a measure.
%}
%\end{remark}

\begin{proof}[Proof of Theorem \ref{mainthm3}]
The compactness property is simply due to the usual apriori $L^2$ bound for the strains $\beta_{\en}\sym$,
while the $\Gamma$-liminf inequality  comes simply by lower semicontinuity.
\par
%Concerning the $\Gamma$-limsup inequality the strategy is  again  to
%approximate a special class of limiting configurations and then proceed by
%density. For instance we could  approximate the strains $\beta$ by $C^1$
%functions, so that their  Curl's  are continuous, and then  perform the
%construction of the recovery sequence analogous to that of Lemma \eqref
%{recovery}.
The main difference with respect to the previous energy regimes is in the proof of the $\Gamma$-limsup
inequality. Again the strategy is  to approximating a special class of limiting configurations and
then to proceed by density, but in this case it will be more convenient to approximate the strains
$\beta$ with $C^1$ functions, so that their  Curl's  are continuous.

Thus,  fix $\beta\in L^2(\Om;\Md)$ such that $\Curl \beta$ is a measure $\mu$ of the type $\mu= g(x) \, dx$
with $g$ continuous and let us construct a sequence
$\{\mu_\e\} \subset X_{\e}(\Om)$   and  a sequence $\beta_\e\in \asmn(\mu_\e)$,
with $ \beta_\e/N_\e\weak \beta$ in $L^2(\Om;\Md)$,  such that
the following $\Gamma$-limsup inequality holds true
\begin{equation}\label{glsasu}
\limsup_{\e\to 0}\frac{1}{N^2_{\e}} \int_\Om W(\beta_\e)\,dx \le \int_\Om W(\beta)\,dx.
\end{equation}

Arguing as in the proof of Lemma \ref{recovery}, it is easy to prove that there exist $C\in \R$
depending only on $\|g\|_{L^\infty(\Om;\R^2)}$ and a sequence of measures
$\mu_\e:=\sum_{i=1}^{M_\e} \xi_{i,\e} \delta_{x_{i,\e}}\in X_\e$, with $|\xi_{i,\e}|\le C$,
such that, setting $r_\e:= C/\sqrt{N_\e}$,
we have $B_{{r_\e}}(x_{i,j}) \in\Om, \,
|x_{i,\e} -x_{j,\e}| \ge 2r_\e$  for every
$x_{i,\e}, \, x_{j,\e}$ in the support set of $\mu_\e$, and, finally,
$$
\frac{\mu_\e}{N_\e} \weakst \mu \text{ in } \ms(\Om;\R^2), \quad
\frac{\hat \mu^{r_\e}_\e}{N_{\e}} \to \mu  \text{ strongly in } H^{-1}(\Om; \R^2),
$$
where $\hat \mu_\e$ is defined according to \eqref{tmr}.
Consider the functions $\hat K_\e^{\mu_\e}$  defined in \eqref{dekeb}
and set
$\bar \beta_\e:= N_\e \beta  + \hat K_\e^{\mu_\e}$.
By \eqref{rotori} we deduce
$$
\Curl \frac{\bar \beta_\e}{ N_\e}\LL \Omega_\e(\mu_\e)= (\mu -\frac{ \hat \mu^{r_\e}_\e}{N_\e}) \to 0 \quad \text{ in }
H^{-1}(\Om; \R^2).
$$
Therefore, we can add to $\bar\beta_\e$ a vanishing sequence $\frac{R_\e}{N_\e} \to 0$ in
$L^{2}(\Om;\Md)$, obtaining the admissible strain
\begin{equation*}
\beta_\e:= (\bar \beta_\e + R_\e) \chi_{\Om_\e(\mu_\e)}.
\end{equation*}
In order to prove that the pair $(\mu_\e,\beta_\e)$ is the desired recovery sequence it is enough
to observe that $\hat K_\e^{\mu_\e}\to 0$ in $L^2(\Om;\Md)$. Indeed, by construction we have $M_\e \le C N_\e$, and therefore
$$
\lim_{\e\to 0} \frac{1}{N_\e^2}\int_{\om_\e} |\hat K_\e^{\mu_\e}|^2 \, dx \le
\lim_{\e\to 0} \frac{C}{N_\e^2}  M_\e |\log \e|
\le \lim_{\e\to 0} C \frac{|\log \e|}{N_\e}  =0.
$$
\end{proof}

\begin{remark}
{\rm Note that in the super-critical regime we can not have a
compactness property for the antisymmetric part of the admissible
strains $\beta_\e$, and indeed it is easy to exhibit  examples where
$\|\beta_\e\sym/N_\e\|_{L^2(\Om;\Mdsym)}\le C$ and
$\|\beta_\e\as/N_\e\|_{L^2(\Om;\Md)}\to \infty$. Note that,  since
we do not have any control on the mass of $\Curl \beta_\e$, we can
not apply Theorem \ref{KTI}. }
\end{remark}

\section*{Acknowledgments}
We want to thank Paolo Cermelli for many fruitful discussions, that have been essential in  formulating
the problem. We also thank Adriano Pisante for stimulating interactions. We thank the Center for Nonlinear Analysis (NSF
Grants No. DMS-0405343 and DMS-0635983) for its support during the preparation
of this paper. The research of G. Leoni was partially supported by the
National Science Foundation under Grants No. DMS-0405423 and DMS-0708039.

\end{document}